\documentclass[11pt]{article}

\usepackage{pgf,tikz,pgfplots}
\pgfplotsset{compat=1.14}
\usepackage{mathrsfs}
\usepackage{tikzit}

\tikzstyle{black dot}=[fill=black, draw=black, shape=circle]
\tikzstyle{small dot}=[fill=black, draw=black, shape=circle, {scale=0.3}]

\tikzstyle{double arrow dotted}=[<->, dashed]

\usetikzlibrary{arrows}
\usepackage[margin=1in]{geometry}
\def\commentsversion{0}  

\ifnum\commentsversion=0 
\newcommand{\comments}[1]{}
\fi

\def\newversion{1}  

\ifnum\newversion=0 
\newcommand{\old}[1]{#1}
\newcommand{\new}[1]{}
\fi
\ifnum\newversion=1 
\newcommand{\old}[1]{}
\newcommand{\new}[1]{\textcolor{blue}{this is a proposed version. let me know if its ok to adopt or not.~~}#1}
\fi

\usepackage{listings}

\ifnum\commentsversion=1 
\newcommand{\comments}[1]{#1}
\fi

\usepackage[T1]{fontenc}    
\usepackage{hyperref}       
\usepackage{url}            
\usepackage{booktabs}       
\usepackage{amsfonts}       
\usepackage{nicefrac}       
\usepackage{microtype}      
\usepackage{amsmath}
\usepackage{bm}
\usepackage{esvect}

\usepackage{amsthm}
\usepackage[noend,linesnumbered]{algorithm2e}
\usepackage{amssymb}
\usepackage{bbm}
\usepackage{enumitem}
\usepackage{graphicx}
\usepackage{caption}
\usepackage{subcaption}
\usepackage{multirow}
\usepackage{floatflt}
\usepackage{color}

\newcommand{\OSS}{\text{OSS }}
\makeatletter
\newtheorem*{rep@theorem}{\rep@title}
\newcommand{\newreptheorem}[2]{%
\newenvironment{rep#1}[1]{%
 \def\rep@title{#2 \ref{##1}}%
 \begin{rep@theorem}}%
 {\end{rep@theorem}}}
\makeatother

\makeatletter
\newtheorem*{rep@proposition}{\rep@title}
\newcommand{\newrepproposition}[2]{%
\newenvironment{rep#1}[1]{%
 \def\rep@title{#2 \ref{##1}}%
 \begin{rep@proposition}}%
 {\end{rep@proposition}}}
\makeatother

\makeatletter
\newtheorem*{rep@lemma}{\rep@title}
\newcommand{\newreplemma}[2]{%
\newenvironment{rep#1}[1]{%
 \def\rep@title{#2 \ref{##1}}%
 \begin{rep@lemma}}%
 {\end{rep@lemma}}}
\makeatother

\newtheorem{definition}{Definition}

\newtheorem{theorem}{Theorem}
\newtheorem{corollary}{Corollary}
\newtheorem{proposition}{Proposition}
\newtheorem{lemma}{Lemma}

\newtheorem{remark}{Remark}

\newreptheorem{theorem}{Theorem}
\newrepproposition{proposition}{Proposition}
\newreplemma{lemma}{Lemma}

        {\hspace*{\fill}$\Box$\par}

\DeclareMathOperator*{\argmax}{arg\,max}

\newcommand{\R}{\mathbb{R}}
\newcommand{\E}{\mathbb{E}}

\newcommand{\I}{\mathcal{I}}
\newcommand{\smooth}{\sigma}

\let\emptyset\varnothing

\title{ 
Beyond Submodular Maximization via One-Sided Smoothness
}

\author{
    Mehrdad Ghadiri\footnote{\texttt{ghadiri@gatech.edu}, Georgia Institute of Technology, Atlanta, GA, USA.} \and
    Richard Santiago\footnote{\texttt{rtorres@ethz.ch}, ETH Zurich, Switzerland.} \and
    Bruce Shepherd\footnote{\texttt{fbrucesh@cs.ubc.ca}, University of British Columbia, Vancouver, Canada.}
}
\date{}

\begin{document}
\maketitle

\begin{abstract}
The multilinear framework was developed to achieve the breakthrough $1-1/e$ approximation for maximizing a monotone submodular function subject to a matroid constraint, which includes the submodular welfare problem as special case. This framework has a continuous optimization part (solving the multilinear extension of a submodular set function) and a rounding part (rounding a fractional solution to an integral one).
We extend  both parts  so that the resulting generalized framework  may be used on a wider array of problems.
In particular, we make a conceptual contribution by  identifying  a family of
parameterized functions and their applications.
As a running example we focus on solving
diversity problems $\max f(S)=\frac{1}{2}\sum_{i,j \in A} A_{ij}: S \in \mathcal{M}$, where $\mathcal{M}$ is matroid. These diversity functions have $A_{ij} \geq 0$ as a measure of dissimilarity of $i,j$, and $A$ has $0$-diagonal. This family of problems ranges from intractable problems such as densest $k$-subgraph, to
$\frac{1}{2}$-approximable metric diversity problems.
The multilinear extension $F$ of such diversity functions satisfies $\nabla^2 F(x)=A \geq 0$ and hence the original multilinear framework (which assumes non-positive Hessians) does not directly apply.  Instead we introduce a new  parameter for functions $F \in {\bf C}^2$ which measures the approximability of the associated problem $\max\{ F(x): x \in P\}$, for solvable downwards-closed polytopes $P$. A function $F$ is called one-sided $\sigma$-smooth if  $\frac{1}{2}u^T\nabla^2 F(x) u \leq \smooth \cdot \frac{||u||_1}{||x||_1} u^T \nabla F(x)$ for all $u,x\geq 0$, $x\neq 0$. 
For $\sigma=0$ this class includes  previously studied classes such as continuous DR-submodular functions, and much more. For the 
 multlinear extension of
a diversity function, we show that it is one-sided $\sigma$-smooth whenever $A_{ij}$ forms a $\sigma$-semi-metric.

We give an $\Omega(1/\sigma)$-approximation for the continuous maximization problem of monotone, normalized one-sided $\sigma$-smooth $F$  with an  additional property: 
 non-positive third order partial derivatives.  Since the multilinear extension of a diversity function has this additional property we can apply the extended multilinear framework to this family of discrete problems. This requires new  matroid rounding techniques for quadratic objectives. The result  is a $\Omega(1/\sigma^{3/2})$-approximation for maximizing a $\sigma$-semi-metric diversity function subject to matroid constraint. This improves upon the previous best bound of $\Omega(1/\sigma^2)$ and we give evidence that it may be tight. For general one-sided smooth functions, we
show the continuous process gives an $\Omega(1/3^{2\sigma})$-approximation, independent of $n$.   In this setting, by discretizing, we present a concrete poly-time algorithm for multilinear functions that satisfy the one-sided $\sigma$-smoothness condition. We also describe a discretization for one-sided smooth functions with $L$-Lipschitz gradients.

\end{abstract}
\newpage

\section{Introduction}
\label{sec:intro}

In a breakthrough result, an optimal $1-1/e$ approximation was given for monotone submodular maximization subject to a matroid constraint  \cite{calinescu2011maximizing,vondrak2008optimal}. 
This resolved a long standing gap between the  best known $1/2$ approximation \cite{fisher1978analysis} and  $1-1/e$ lower bound \cite{feige1998threshold}. It also provides a tight approximation for the submodular welfare problem \cite{feige2006approximation}. A key insight   was to use a   continuous relaxation based on the multilinear extension (ME) $F$ of a set function $f:2^{[n]} \rightarrow \mathbb{R}_{\geq 0}$.
For $x \in [0,1]^n$, $F(x)$ is defined as $E[f(R(x))]$, where
$R(x)$ is a random set with each $i$ being selected independently with probability $x_i$. In particular, for  $S \subseteq [n]$, $F(\mathbbm{1}_S)=f(S)$. Thus a valid {\em multilinear relaxation} for a discrete problem $\max \{ f(S): S \in \mathcal{M}\}$ is obtained: $\max \{F(x): x \in P_{\mathcal{M}}$,
where $P_{\mathcal{M}}=conv(\mathbbm{1}_S: S \in \mathcal{M})\}$. 

This  framework has inspired a successful stream of research including 
on non-monotone submodular functions \cite{buchbinder2019constrained}   and new `contention resolution' rounding techniques for general polytopes  \cite{chekuri2014submodular}.
In this work we make a conceptual contribution by identifying a family of parameterized set functions where an extension to the multilinear framework can be brought to bear.  We also give several applications for this generalized framework.

Using this framework  to solve a discrete problem requires two essential ingredients. First, algorithmic tools to find a good solution $x^*$ for the multilinear relaxation. Second, to be able to convert a solution $x^*$ into a set $S$ with $f(S) \approx F(x^*)$. As multilinear extensions are neither concave nor convex, it is not a priori clear that 
the fractional problem itself would be tractable. For monotone submodular functions, however,   a gradient-based  technique\textemdash called {\em continuous greedy}\textemdash is shown to provide a $1-1/e$ approximation \cite{vondrak2008optimal}.   This analysis relies on
the fact that MEs of submodular functions have 
{\em non-positive  second derivatives}.  
Functions $F \in {\bf C}^2$  with this property
  are called  {\em  continuous DR-submodular} \cite{bian2017guaranteedI} (cf. \cite{soma2015generalization}). 
 The rounding step for matroids relies on a different property of  multilinear extensions. Namely, if $f$ is submodular, then $F$ is convex in any direction $\bm{e}_i-\bm{e}_j$, where $\bm{e}_i$ denotes the characteristic vector of $\{i\}$. This allows a lossless conversion to a discrete solution, for instance, by using pipage rounding. The combination of the fractional algorithm and rounding provides the $1-1/e$ approximation.


In this paper, we develop a wider  scope for the multilinear framework
and give evidence of its use   in other applications.  
 One motivating  example  is  {\em diversity maximization}  \cite{lin2010multi,kulesza2011k,wang2008multi} which has  applications in machine learning \cite{ZadehGMZ17,ghadiri2019distributed}, document aggregation \cite{AbbassiMT13}, web search \cite{RadlinskiD06}, recommender systems \cite{XinCYH06,CarbonellG98}, and many more. 
One widely used model is $\max\{ f(S): |S| \leq k\}$,  where $f(S)=\frac{1}{2}\sum_{i,j\in S} A_{ij}$ for all $S\subseteq [n]$.
We refer to $f$ as a {\em diversity function} if
$A \geq 0$ is symmetric and has  $0$-diagonal. We think of   $A_{ij}$ as measuring  dissimilarity between items $i,j$.  
This family of (supermodular) maximization problems ranges from challenging examples such as {\em $k$-densest subgraph}, with
best known approximation  $\Omega(1/n^{0.25+\epsilon})$ \cite{bhaskara2012polynomial,manurangsi2017almost},  to   {\em metric diversity} ($A_{ij}$ forms a metric) which is  $\frac{1}{2}$-approximable \cite{hassin1994notes, BorodinLY12}.
Since the multilinear extension of a diversity function $f$ has  Hessian $A$ which is  non-negative (as opposed to non-positive),  the standard $(1-1/e)$-approximation from continuous greedy does not directly apply. One of our main messages is that the metric property in diversity maximization is intrinsic to the tractability of the multilinear relaxation.


To describe  our extended  multilinear framework we first discuss  the fractional problem and later discuss rounding. We introduce a parameterized family of monotone, non-negative functions  $F \in {\bf C}^2$ and then show that  the parameter
governs the approximability of the  problem   $\max \{F(x): x \in P\}$, for downwards-closed polytopes $P$.  To achieve this we cannot directly rely on a crucial fact used in the analysis  of the continuous greedy  process: that the rate of change of $F$ at a point $x$ is at least the current {\em deficit}, defined as {\sc OPT}$-F(x)$. For MEs of submodular functions, this follows from  $F$ being concave in non-negative directions, which  in turn relies on  $F$'s second derivatives being non-positive.  Instead, 
we define a family  of   functions  whose growth in non-negative directions is constrained by a  parameter $\sigma$.
A  function $F:\mathbb{R}_{\geq 0}^n\rightarrow \mathbb{R}$ in ${\bf C}^2$ is called \emph{one-sided $\sigma$-smooth} (or $\sigma$-\OSS  for short)  if it satisfies 
\begin{equation}
\label{eqn:OSS} \tag{OSS}
\frac{1}{2}u^T\nabla^2 F(x) u \leq \smooth \cdot \frac{||u||_1}{||x||_1} u^T \nabla F(x),
\end{equation}

\noindent
for all $u,x\geq 0$, $x\neq 0$. 

The class of $0$-\OSS functions already contains  interesting and familiar functions. This includes the continuous DR-submodular  functions, as their  Hessians are non-positive \cite{bian2017guaranteedI,bian2017continuousII}; the DR-submodular form a superset of the 
functions originally considered for continuous greedy \cite{vondrak2008optimal}. The $0$-OSS functions contain much more however, as we discuss later in Section~\ref{sec:zero}.  For all of these functions,
the continuous greedy process  returns a solution within $1-1/e$ of the optimum; in some cases, converting this into  a concrete polytime algorithm  requires additional assumptions.

For larger values of $\sigma$, one example of $\sigma$-smooth functions is the class of $\sigma$-semi-metric diversity functions. Namely,
the parameter $\sigma$  corresponds to the matrix  $A$ being a {\em $\sigma$-semi-metric}.
This means that $A_{ik} \leq \sigma (A_{ij} + A_{jk})$ for all $i,j,k$, see Proposition~\ref{prop:ossdiscretequad} in Appendix~\ref{focsapp:semimetricOSS}. This captures diversity functions addressed in the literature, such as metric diversity \cite{BorodinLY14weak} ($\sigma=1$), and   negative-type distances \cite{CevallosEZ17,CevallosEZ16} or  Jensen-Shannon divergence which has been used to measure dissimilarity of probability measures (both have smoothness $\sigma=2$), see Appendix~\ref{focsapp:semimetricOSS}.

Our main contribution to the fractional problem is to show that
one-sided smoothness of a monotone, non-negative function governs the approximability of    $\max \{F(x): x \in P\}$, for downwards closed polytopes $P$. This is
reminiscent of how Lipschitz smoothness bounds convergence rates in convex optimization --- see Appendix~\ref{focsapp:lipschitz} for a discussion about the difference of Lipschitz smoothness and OSS. If $F$ additionally has non-positive third order partials, then we show that continuous greedy can be adapted to become a $\Omega(1/\sigma)$-approximation, and we show this is tight.  This class includes the discussed MEs for diversity maximization. 
We can combine this with new rounding techniques to obtain
unified results for maximizing diversity functions over matroids. 
Unlike for submodular functions, this requires the best-of-two rounding methods.  One is inspired by swap rounding,  previously applied to the submodular case. The other extends the  approximate integer decomposition framework \cite{chekuri2009approximate} to handle the ``pairwise  terms'' in  diversity functions.

For general  $\sigma$-smooth functions, without any third order assumption,
we can  obtain an $\Omega(1/3^{2\sigma})$ approximation (independent of $n$) for the continuous greedy process.
We can no longer use the 2nd order Taylor Polynomial since
we do not have non-positivity of 
 the third order error term. Instead we work with
 the 1st Order Taylor expansion but this requires
a new  upper bound on $u^T \nabla F(x+\epsilon u)$, the directional derivative, in a neighbhourhood of $x$.  In the fully  general setting we  need a  (strong) lower bound on $u^T \nabla F(x+\epsilon u)$ to make a concrete algorithm.
 However,  for \emph{multilinear} $\sigma$-OSS functions a polytime algorithm is shown independent of any additional assumptions. We also consider discretization
 for one-sided smooth functions with Lipschitz s gradients, and for a class of $0$-smooth functions which are not continuous DR submodular Section~\ref{app:procure}.
 
 \subsection{The Zero One-Sided Smooth Class}
 \label{sec:zero}
 
The class of $0$-\OSS functions is interesting in its own right. For monotone, non-negative members of this family our results show that the continuous greedy process yields a $1-1/e$ approximation for the fractional problem. Obtaining a polytime algorithm (discretization) is not immediate but we can establish natural conditions on cases where this can be achieved.    The general $0$-\OSS family forms  a very broad class of functions. For instance, it
contains every concave function $F \in {\bf C}^2$ (even though our results are only tailored  for the monotone, non-negative functions in this class).  This means it also   contains the continuous DR-submodular functions (Hessians are non-positive). This containment is proper since 
 there are $0$-\OSS functions with positive off-diagonal entries in their Hessian. 
It is interesting to compare with the related family of continuous submodular functions
that has been developed in the context of minimization  \cite{bach2019submodular}. Continuous submodular functions are defined as  having Hessians with non-positive off-diagonal entries,  but they {\em may} have positive diagonal entries. In contrast,  $0$-\OSS {\em must} have non-positive diagonals but {\em may} have positive off-diagonals (cf. Appendix~\ref{app:procure}) - see Figure~\ref{fig:zerosmooth}.

The (general) $0$-\OSS family can be defined as the functions $F \in {\bf C}^2$ for which $- \nabla^2 F(x)$ is copositive for every $x$\footnote{A matrix $A$ is {\em copositive} if $u^T A u \geq 0$ for every $u \geq 0$ \cite{dur2010copositive}.}. While recognition of copositive matrices is NP-hard \cite{murty1985some}, we  propose a  strategic procurement problem which is modelled as maximizing a 
 quadratic functions $F(x)=\frac{1}{2} x^T (-A) x + b^Tx$ where $A$ is a copositive matrix  defined by the user. Note that this family of objectives are a generalization   of concave quadratics.
We refer to the resulting (fractional)  maximization problem as  {\em  diversified procurement} discussed in Appendix~\ref{app:procure}. 

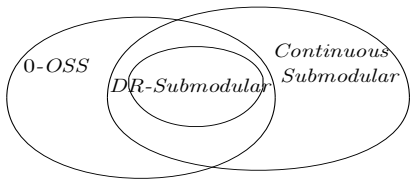
\begin{figure}[thbp]
\begin{center}
\resizebox{2.2in}{1.2in}{
\begin{tikzpicture}[scale=0.2]
	\begin{pgfonlayer}{nodelayer}
		\node [style=none] (0) at (-18, -1) {};
		\node [style=none] (1) at (4, -1) {};
		\node [style=none] (2) at (-9.75, -1) {};
		\node [style=none] (3) at (15, -1) {};
		\node [style=none] (5) at (-3, 0.25) {{\em DR-Submodular}};
		\node [style=none] (7) at (-14, 2.25) {{\em $0$-OSS}};
		\node [style=none] (8) at (8.5, 3.75) {{\em Continuous}};
		\node [style=none] (9) at (9.25, 1.25) {{\em Submodular}};
		\node [style=none] (10) at (-8, 0) {};
		\node [style=none] (11) at (3, 0) {};
		\node [style=none] (12) at (-8, 1) {};
		\node [style=none] (13) at (3, 1) {};
	\end{pgfonlayer}
	\begin{pgfonlayer}{edgelayer}
		\draw [bend left=90, looseness=1.25] (2.center) to (3.center);
		\draw [in=90, out=90, looseness=1.25] (0.center) to (1.center);
		\draw [bend right=90, looseness=1.25] (0.center) to (1.center);
		\draw [bend right=90] (2.center) to (3.center);
		\draw [bend right=75, looseness=1.25] (10.center) to (11.center);
		\draw (10.center) to (12.center);
		\draw [bend left=75] (12.center) to (13.center);
		\draw (13.center) to (11.center);
	\end{pgfonlayer}
\end{tikzpicture}
}
\caption{Venn diagram of $0$-OSS, continuous DR-submodular, and continuous submodular functions.}
\label{fig:zerosmooth}
\end{center}
\end{figure}

\section{Our Results}
\label{sec:ourResult}

Our results are of three types:  1) fractional approximations, 2) rounding, and  3) hardness results. These are presented in Sections \ref{sec:frac-approx}, \ref{sec:intgap}, and \ref{sec:hardness} respectively. In all of our results, we assume that the function is monotone and non-negative. 

\textbf{Fractional approximation.} Our main result in this part (Theorem \ref{thm:greedyBound}) shows that a modified version of the continuous greedy process gives a $(1-\exp{(-(1-\alpha)(\frac{\alpha}{\alpha+1})\strut\strut^{2\sigma})})$-approximation for maximizing a non-negative, monotone, $\sigma$-\OSS function subject to a downwards-closed polytope, where $\alpha$ is an arbitrary number in $[0,1)$. We remark that for  $\sigma=0$, our results recover the $(1-1/e)$-approximation  \cite{vondrak2008optimal,calinescu2011maximizing} for maximizing the multilinear extension of a submodular function, by setting $\alpha=0$. For $\alpha=0.5$, our approximation is better than $\frac{0.5}{3^{2\sigma}+0.5}$. For fixed $\sigma$ this gives  a constant-factor approximation independent of $n$. At present, we do not know the correct  dependence on $\sigma$.  However,
the dependence improves to linear with an additional assumption that third-order partials are non-positive.  
More precisely,  we obtain a $(1-\exp{(-\frac{1}{4\sigma+2})})$-approximation; see Theorem \ref{thm:thirdderiv}. As mentioned in Section~\ref{sec:intro} this gives  a $\Omega(1/\sigma)$ approximation which  is in fact tight within a  constant factor (cf. Corollary~\ref{cor:hardness-dks} discussed below). 
One example of such  functions are  multilinear extensions of semi-metric diversity functions, i.e., whose Hessian is a $\sigma$-semi-metric (discussed further in the rounding part). 

The `algorithm' described in the previous paragraph is a continuous-time process, and it is not immediately obvious that it can be implemented as a discrete algorithm. Some readers may wish to take it on faith that this is possible and skip ahead to the rounding results. There are actually some subtleties involved which require two distinct approaches. 
One of our methods works for multilinear functions, while the other works for general \OSS functions but needs an additional parameter that governs the growth of the first order derivatives from below. A fuller discussion is  in
 Appendix \ref{focsapp:discretization}.

 \textbf{Rounding.} In this part, we consider  maximizing  set functions of the form
 \[f(S)=\frac{1}{2} \sum_{u,v \in S} A(u,v) + \sum_{v \in S} b(v)\]
 \noindent
 where $A$ is a symmetric matrix with $0$-diagonal. If  $A,b \geq 0$,
 then these are the previously discussed diversity functions, but more generally we refer to these as  {\em discrete quadratics} (aka second-order modular \cite{KorulaMZ15}) as their  extensions are quadratic functions $F(x)=\frac{1}{2}x^T A x+b^T x$.  Since their third derivatives are obviously $0$, they `qualify' for the $\Omega(1/\sigma)$-approximation  from the preceding section.
 Hence we have a $\Omega(1/\sigma)$-approximation for maximizing $F$ over a matroid polytope $P_{\mathcal{M}}$ when $F$ is $\sigma$-\OSS (i.e., $A$ is $\sigma$-semi-metric). In order to solve the discrete problem $\max \{f(S): S \in \mathcal{M}\}$, we need to transform this fractional solution to a discrete one.


 
  We present two different rounding procedures which combined lead to a rounding gap of $O(\min\{ \frac{r}{c-2}, 1+\frac{\sigma}{r}\})$, where $r$ denotes the rank of the matroid and $c$ the size of a smallest circuit.
  Surprisingly we show this is tight (see hardness part). 
  Moreover, this  yields   an $O(\sqrt{\sigma})$ rounding gap independent of $r$ and $c$ (Theorem \ref{thm:integrality-gaps}). Combining the modified continuous greedy algorithm with our rounding result, there is an $\Omega(1/\sigma^{3/2})$-approximation for maximizing $\sigma$-semi-metric diversity functions subject to a matroid constraint (Theorem~\ref{thm:approx-sigma-diversity}). This  improves  the best known $\Omega(1/\sigma^2)$ bound \cite{ZadehG15}. In addition, we note that $\Omega(1/\sigma^{3/2})$ is a pessimistic bound in general. For instance, for uniform matroids we have $c=r+1$, which leads to an $O(1)$ rounding gap and hence an improved $\Omega(1/\sigma)$; as discussed below  this is actually tight.

This $O(1)$ rounding gap implies that for a cardinality constraint, the approximation bound of the discrete problem is asymptotically the same as the bound for the continuous problem. Thus the continuous problem of maximizing a general multilinear quadratic function over the simplex $||x||_1\leq k$, is as hard as solving the densest $k$-subgraph problem (see Corollary \ref{cor:hardness-dks}). This is similar to
the situation for continuous maximization of MEs of submodular functions. Such continuous hardness problems have received less attention, as remarked by  De Klerk \cite{de2008complexity}: ``approximation algorithms have been studied extensively for combinatorial optimization
problems, but have not received the same attention for NP-hard {\em continuous}
optimization problems.''  We close this section by discussing our hardness results for the discrete problems.

\textbf{Hardness.} In this part, we show that the hardness of approximation is also governed by the smoothness parameter of the function. More specifically, in Theorem \ref{thm:approxhardness} we show that assuming the planted clique conjecture, for a constant $\sigma$ it is hard to approximate the maximum of a $\sigma$-semi-metric diversity function subject to a cardinality constraint within a factor better than $2\sigma$. We also show that for a super constant $\sigma$, it is hard to find any constant factor approximation.

In Theorem \ref{thm:intgaplower} we give a lower bound of $\Omega(\min \{ \frac{r}{c-2},\frac{\sigma}{r}\})$ for the rounding gap of a $\sigma$-semi-metric diversity function over a matroid polytope. This shows that our rounding methods are essentially tight. In particular,  each step of our algorithm for maximizing diversity functions (i.e., maximizing the continuous function and rounding) is  tight. This leads us to speculate that the $\Omega(1/\sigma^{3/2})$-approximation (Theorem~\ref{thm:approx-sigma-diversity}) is asymptotically tight.

\section{Related Work}

We first discuss work related to solving the continuous problem in the multilinear framework.
Other adaptations of the  continuous greedy algorithm 
have been developed for applications to non-monotone submodular maximization \cite{feldman2011unified,ene2016constrained} and distributed  maximization \cite{barbosa2016new}. Another avenue aimed to generalize the class of ${\bf C}^2$ functions originally considered
\cite{vondrak2008optimal}. For instance Bach \cite{bach2019submodular} develops minimization algorithms for the family of {\em continuous submodular} functions \cite{lorentz1953inequality}  defined on compact product subsets of $\mathbb{R}^n$. A function $F \in {\bf C}^2$ is submodular if the off-diagonal entries of its Hessian are non-positive. This class is an extension of lattice submodular functions \cite{topkis1978minimizing,fujishige2005submodular} (a lattice is a poset closed under meet and join operations and hence these generalize submodular set functions). DR-submodularity is a restricted form of lattice submodular functions  introduced for maximization \cite{soma2014optimal,soma2015generalization}. These
generalize to continuous DR-submodular functions in ${\bf C}^2$ for which all entries of the Hessian are non-positive \cite{bian2017guaranteedI}. 
 The continuous greedy algorithm has also been studied
 for  maximization of these continuous functions.
 Discretization requires  an additional bound on Lipschitz smoothness and then  a $(1-1/e)$-approximation can be achieved as step sizes approach $0$. This is done
 for maximizing a (monotone and non-monotone)   DR-submodular function over a downwards-closed polytope \cite{bian2017guaranteedI,bian2017continuousII}. 
This is introduced  as an alternative to multilinear extension which is  more practical  to evaluate,  and for which a gradient-based algorithm leads to a $1/4$ approximation over downwards-closed polytopes \cite{gillenwater2012near}.

As for  discrete problems, after the introduction of
the multilinear framework,  there have been many developments. One highlight is the introduction  of contention resolution schemes \cite{chekuri2009approximate} which   allow one to work with more general polytopes. An online version of this approach has also been developed \cite{feldman2016online} with  applications in algorithmic game theory.
 In a different direction, the work of \cite{filmus2014monotone} gives a combinatorial local search $1-1/e$ approximation algorithm for maximizing monotone submodular functions over matroids.

The
diversity maximization problem $\max \{ f(S): |S| \leq k \}$ has proved extremely versatile for many applications, as noted in Section~\ref{sec:intro}. 
On the algorithmic side, a greedy $\frac{1}{2}$-approximation was
devised in \cite{hassin1994notes} and this was generalized to matroid constraints \cite{borodin2015proportionally}. The latter  was extended  to yield a $\Omega(1/\sigma^2)$-approximation
whenever the diversity costs $A_{ij}$ form a $\sigma$-semi-metric \cite{ZadehGMZ17}. That is, $A_{ik} \leq \sigma (A_{ij} + A_{jk})$ for all $i,j,k$. A PTAS has
also been developed when the $A_{ij}$'s are negative type distances \cite{CevallosEZ17,CevallosEZ16}. 

There is also work that extends set function maximization beyond submodularity. In   \cite{bian2017guaranteesIII} 
a greedy algorithm is shown to give good approximations
for  a  family of set functions which are parameterized by curvature and submodularity ratio values. In \cite{BorodinLY12}, a $\frac{1}{2}$-approximation is developed for the problem of maximizing the sum of a submodular function and a metric diversity function. A generalization of this function, called proportionally submodular functions, is considered in \cite{borodin2015proportionally}.
Another extension is to maximization of weakly submodular functions where non-negativity of the function is relaxed   \cite{harshaw2019submodular}.

\section{Fractional Approximation.}
\label{sec:frac-approx}

In this section, we first discuss a key property of one-sided smooth functions, which is the main tool in our analysis. This property asserts that for a point $x$, the directional derivative at points close to $x$ is bounded by a factor of the directional derivative at $x$. 

We then present a variant of the continuous greedy process which we use for both general $\sigma$-OSS functions and those that have non-positive third-order partial derivatives. We analyze this algorithm for both classes of functions. The discretization of the continuous greedy process is discussed in Appendix~\ref{focsapp:discretization}.

\subsection{Notations}

We use $\{\bm{e}_1,\ldots,\bm{e}_n\}$ to denote the standard basis of $\mathbb{R}^n$ and $[n]:=\{1,\ldots,n\}$ to refer to the ground set of a set function. We denote the $i$'th coordinate of a vector $x$ with $x_i$.
For a set $R\subseteq [n]$, we denote by $\mathbbm{1}_R$ its characteristic vector. Given a vector $x$ we denote its support by $supp(x)$, i.e., the set of non-zero coordinates of $x$. For a matrix $A$, we use $A_{ij}$ and $A(i,j)$ interchangeably to refer to the $i,j$ entry of $A$.

\subsection{A Key Property of One-Sided Smoothness}

The following result describes a property of one-sided smoothness that plays a key role in the analysis of the algorithm. It enables us to bound the first order Taylor's polynomial of the function.

\begin{lemma} 
Let $x \in [0,1]^n\setminus\{\vec{0}\}$, $u \in [0,1]^n$ and $\epsilon >0$ such that $x+\epsilon u \in [0,1]^n$. Let $F: [0,1]^n \rightarrow \mathbb{R}$ be a non-negative, monotone function which is one-sided $\sigma$-smooth on $\{y|x+\epsilon u \geq y\geq x\}$.
Then 
\[
u^T \nabla F(x+\epsilon u) \leq \left(\frac{||x+\epsilon u||_1}{||x||_1}\right)^{2\sigma} (u^T \nabla F(x)).
\]
\label{lemma:epsilonchangegradient}
\end{lemma}
\begin{proof}
Let $g(t):=u^T \nabla F(x+tu)$. By the Chain Rule we have $g'(t)=u^T \nabla ^2 F(x+tu) u$.

\noindent
By one-sided $\sigma$-smoothness on $\{y|x+\epsilon u \geq y\geq x\}$, for any $0\leq t\leq \epsilon$,
\[
g'(t)=u^T \nabla^2 F(x+t u) u \leq 2\sigma \frac{||u||_1}{||x+tu||_1} u^T \nabla F(x+tu) =2\sigma \frac{||u||_1}{||x+tu||_1} g(t) \leq 2\sigma \frac{||u||_1}{||x+tu||_1} (g(t)+h),
\]

\noindent
for any $h>0$. Therefore, using that $g(t)+h > 0$ for all $t$ (since $g(t)\geq 0$), we have
\begin{equation}
    \label{eqn:pre-int}
    \frac{g'(t)}{g(t)+h}\leq 2\sigma \frac{||u||_1}{||x+tu||_1}.
\end{equation}

\noindent
We integrate both sides of (\ref{eqn:pre-int}) with respect to $t$. On the left hand side we get
\[
\int_{0}^{\epsilon} \frac{g'(t)}{g(t)+h} dt = \ln (g(t)+h) \biggl|_{0}^\epsilon = \ln(\frac{g(\epsilon)+h}{g(0)+h}),
\]

\noindent
and on the right hand side we get
\[
2\sigma \int_{0}^{\epsilon} \frac{||u||_1}{||x+tu||_1} dt = 2\sigma \ln(||x+tu||_1)\biggl|_{0}^\epsilon = 2\sigma \ln (\frac{||x+\epsilon u||_1}{||x||_1}),
\]

\noindent
where we use that $||u||_1 = \sum_i u_i = \frac{d}{dt} \sum_i (x_i + t u_i) = \frac{d}{dt} ||x+tu||_1$.

Therefore
$\ln (\frac{g(\epsilon)+h}{g(0)+h}) \leq \sigma \ln (\frac{||x+\epsilon u||_1}{||x||_1})$, and hence 
$
g(\epsilon)+h\leq \left(\frac{||x+\epsilon u||_1}{||x||_1}\right)^{2\sigma} (g(0)+h).
$
\noindent
Since this holds for any $h>0$ taking the limit yields the desired result.
\end{proof}

\subsection{Continuous Greedy and One-Sided $\sigma$-Smoothness}
\label{sec:cont-greedy-and-smoothness}
We now provide an adaptation of the {\em continuous greedy algorithm}, originally introduced in \cite{vondrak2008optimal}. 
Algorithm~\ref{alg:ctsgreedy} is for maximizing a monotone $\sigma$-OSS function over a polytime separable downward-close polytope.
Unlike the classical continuous greedy, our algorithm starts from a non-zero point, which allows us to take advantage of Lemma~\ref{lemma:epsilonchangegradient}. Because of this, we call our algorithm \emph{jump-start continuous greedy}.

\RestyleAlgo{algoruled}
\begin{algorithm}[t]
\footnotesize
\textbf{Input:}  A monotone $\sigma$-OSS function $F:[0,1]^n \to \mathbb{R}_{\geq 0}$, a polytime separable downward-closed polytope $P$, and  $\alpha \in [0,1)$ \\[0.6ex]
$v^* \leftarrow \argmax_{x \in P} ||x||_1$ \\
$x(0) \leftarrow \alpha v^*$ \\
$v_{max}(x) \leftarrow \argmax_{v \in P} \{ v^T \nabla F(x) \}$ \\
\For{$t \in [0,1]$} {
	Solve  $x'(t) =(1-\alpha)  v_{max}(x(t))$ ~with boundary condition $x(0) = \alpha v^*$ \\ 
}
\Return $x(1)$ \;
\caption{Jump-Start Continuous Greedy}
\label{alg:ctsgreedy}
\end{algorithm}

\begin{theorem}
\label{thm:greedyBound}
	Let $F:[0,1]^n \to \mathbb{R}_{\geq 0}$ be a monotone $\sigma$-OSS function. Let $\alpha \in [0,1)$ and
	 $P$ be a polytime separable, downward-closed, polytope. If we run the jump-start continuous greedy process (Algorithm~\ref{alg:ctsgreedy}) then $x(1) \in P$ and
     $
     F(x(1)) \geq  [1-\exp{(-(1-\alpha)(\frac{\alpha}{\alpha+1})\strut\strut^{2\sigma})}]\cdot OPT
     $
where $OPT := \max \{F(x):x\in P\}$.
\end{theorem}
\begin{proof}
 The main idea of the proof is to show that moving in the $v_{max}$ direction guarantees a fractional progress equal to $(\frac{\alpha}{\alpha+1})^{2\sigma} (OPT - F(x))$.
 Let $x^* \in P$ be such that $F(x^*)=OPT$. Also, let $x\in \{x(t):0\leq t\leq 1\}$ and $u=(x^*-x) \vee 0$, i.e.,  $x^* \vee x = x + u$ (where $\vee$ denotes the component-wise maximum operation).
   We have by Taylor's Theorem that for some $\epsilon \in [0,1]$:
    \[
    OPT \leq F(x^* \vee x) = F(x) + u^T \nabla F(x+\epsilon u) \leq F(x) + \left(\frac{||x+\epsilon u||_1}{||x||_1} \right)^{2\sigma}  u^T  \nabla F(x), 
\]
where the last inequality follows from Lemma~\ref{lemma:epsilonchangegradient}.
By the choice of $x(0)$ we have that $||x(0)||_1 \geq \alpha ||w||_1$ for any $w \in P$, and then since $u \in P$ and $x(t)$ is non-decreasing in each component (because $v_{max}$ is always non-negative) we also have 
    \[\frac{||x+\epsilon u||_1}{||x||_1}
    \leq \frac{||x+u||_1}{||x||_1} = 1+\frac{||u||_1}{||x||_1} \leq 1+\frac{||u||_1}{||x(0)||_1} \leq 1+\frac{1}{\alpha} = \frac{\alpha+1}{\alpha}.\]
    By the choice of $v_{max}$ and above inequalities it follows that for any $x\in\{x(t):0\leq t\leq 1 \}$,
    \[
    v_{max}(x)\cdot \nabla F(x) \geq u^T \nabla F(x) \geq \frac{1}{ \left(\frac{||x+\epsilon u||_1}{||x||_1} \right)^{2\sigma}} (OPT - F(x)) \geq (\frac{\alpha}{\alpha+1})^{2\sigma} (OPT-F(x)).
    \]
    Let $\rho=(\frac{\alpha}{\alpha+1})^{2\sigma}$. Then using chain rule, we have
	\begin{equation*}
		\frac{d}{dt} F(x(t)) = \nabla F(x(t)) \cdot x'(t) = \nabla F(x(t)) \cdot (1-\alpha) v_{max}(x(t)) \geq\rho (1-\alpha) [OPT - F(x(t))].
	\end{equation*}
	We solve the above differential inequality by multiplying by $e^{\rho (1-\alpha) t}$.
	\begin{eqnarray*}
		\frac{d}{dt} [e^{\rho (1-\alpha)t}\cdot F(x(t)) ] & = & \rho (1-\alpha)e^{\rho (1-\alpha)t}\cdot F(x(t)) + e^{\rho (1-\alpha)t}\cdot \frac{d}{dt} F(x(t)) \\
		 & \geq &  \rho (1-\alpha)e^{\rho (1-\alpha)t}\cdot F(x(t)) + \rho \cdot e^{\rho (1-\alpha)t} (1-\alpha) [OPT - F(x(t))] \\
		 & = & \rho (1-\alpha)e^{\rho (1-\alpha)t}\cdot OPT.
	\end{eqnarray*}
	
	Integrating the LHS and RHS of the above equation between $0$ and $t$ we get
	\begin{eqnarray*}	
		e^{\rho (1-\alpha)t}\cdot F(x(t)) - e^0 \cdot F(x(0)) & \geq & \rho (1-\alpha)OPT \int_{0}^{t} e^{\rho (1-\alpha)\tau} d\tau \\ &=& 
		\rho (1-\alpha)OPT \cdot [\frac{e^{\rho (1-\alpha)t}}{\rho (1-\alpha)} - \frac{1}{\rho (1-\alpha)}]
		 = OPT \cdot [e^{\rho (1-\alpha)t}-1].
	\end{eqnarray*}
	Hence
	\begin{eqnarray*}	
	F(x(t))  & \geq & [1-\frac{1}{e^{\rho (1-\alpha)t}}] OPT + \frac{F(x(0))}{e^{\rho (1-\alpha)t}} \geq [1-\frac{1}{e^{\rho (1-\alpha)t}}] OPT,
	\end{eqnarray*}
	where the last inequality follows from the fact that $F$ is non-negative.
    Substituting $t=1$ and $\rho=(\frac{\alpha}{\alpha+1})^{2\sigma}$ gives the desired result.
\end{proof}
    
In Proposition~\ref{prop:bestc} in Appendix~\ref{focsapp:contgreedy} we provide an explicit expression for the best value of $\alpha$ (in terms of $\sigma$) for Algorithm~\ref{alg:ctsgreedy} when we are dealing with $\sigma$-OSS functions. 

As discussed in Section~\ref{sec:ourResult}, if the third-order partial derivatives of $F$ are non-postive, then the approximation factor of Algorithm~\ref{alg:ctsgreedy} improves to $\Omega(1/\sigma)$.

\begin{theorem}
\label{thm:thirdderiv}
	Let $F:[0,1]^n \to \mathbb{R}_{\geq 0}$ be a monotone $\sigma$-OSS function with non-positive third-order partial derivatives.  Let $P$ be a polytime separable, downward-closed, polytope. If we run the jump-start continuous greedy process (Algorithm~\ref{alg:ctsgreedy}) with $\alpha=1/2$, then $x(1) \in P$ and
$F(x(1)) \geq  [1-\exp{(-\frac{1}{4\sigma+2})}]\cdot OPT \geq  \frac{1}{4\sigma+3} \cdot OPT$, where  $OPT := \max \{F(x):x\in P\}$.
\end{theorem}

The main idea for proving Theorem~\ref{thm:thirdderiv} is to use the third-order Taylor's polynomial and use the non-positivity of third-order partials and the defining property of $\sigma$-OSS functions. More specifically, because the third-order partials are non-positive, we have
\[
OPT\leq F(x^* \vee x) \leq F(x) + u^T \nabla F(x) + \frac{1}{2} u^T \nabla^2 F(x) u \leq F(x) + (1+\sigma \cdot \frac{||u||_1}{||x||_1}) u^T \nabla F(x)
\]
Then using the fact that $||x||_1$ is large (because we start from a non-zero point), we can conclude that 
$
v_{max}(x)\cdot \nabla F(x) \geq  \Big(\frac{\alpha}{\alpha+\sigma}\Big)  \Big(OPT - F(x) \Big)
$.
This inequality is then used to derive the desired result. For details of the proof of Theorem \ref{thm:thirdderiv}, see Appendix~\ref{focsapp:contgreedy}.

Algorithm~\ref{alg:ctsgreedy} is a continuous process and in general, it cannot be implemented in finite time. Therefore, we give a discretization of this process. In Appendix~\ref{focsapp:discretization}, we show that starting from $x^0=\alpha v^*$ and using the update rule $x^{t+\delta} = x^t + \delta \cdot (1-\alpha) \cdot v_{max}(x^t)$
with the appropriate step size $\delta$, we can recapture similar approximation factors. We present different results for the discretization which are very similar in nature. The first one asserts that, if $F$ is the multilinear extension of some set function $f$, then using $\delta=O(1/n^3)$, the output of the discrete algorithm satisfies
$
F(x^1)\geq (1-\exp(-\frac{1}{2}(1-\alpha)(\frac{\alpha}{\alpha+1})^{2\sigma}))(1-o(1)) OPT
$.
See Theorem~\ref{thm:discretizemultilinear} in Appendix~\ref{focsapp:discretization}.

The second result states that for a function $F$ that satisfies $u^T \nabla F(x+ \epsilon u) \geq \beta u^T \nabla F(x)$ for all $u,x\in P$ and $\epsilon\in [0,1]$, using $\delta=1/\beta n$, the output of the discrete algorithm satisfies
$F(x^{1}) \geq  \Big( 1-\exp{(- \beta (1-\alpha)(\frac{\alpha}{\alpha+1})\strut\strut^{2\sigma}} \Big) OPT$. See Theorem~\ref{thm:discretizebetageneral} in Appendix~\ref{focsapp:discretization}.
Note that, for example, the functions with a non-negative Hessian satisfy the mentioned inequality with $\beta=1$.


\section{Rounding}
\label{sec:intgap}

Let $\mathcal{M}=([n],\mathcal{I})$ be a matroid and $P_{\mathcal{M}}$ be its polytope.
In this section we study the integrality gap for a quadratic program:  $\max F(x): x \in P_{\mathcal{M}}$.
Here $F$ is a non-negative, quadratic multilinear function $F(x) =  \frac{1}{2} x^T A x + b^Tx$ such that $A,b \geq 0$ and $A$ is a symmetric, zero diagonal matrix. 

There are unbounded gaps for such quadratic programmes  even for graphic matroids if we allow parallel edges (see Theorem \ref{thm:intgaplower}).
Fortunately these large gaps transpire for a simple reason, namely when the matroids have very small circuits. We are able to obtain 
 the following integrality gap upper bound.

\begin{theorem}[Quadratic Integrality Gap over Matroids]
\label{thm:integrality-gaps}
Let $f$ 
be a set function whose   multilinear extension $F$ is $\sigma$-\OSS.
Let $\mathcal{M}$ be a matroid of rank $r$, minimum circuit size $c$, and matroid polytope $P_\mathcal{M}$.
Then there is a polytime algorithm which given $x^* \in P_\mathcal{M}$ produces an integral vector $\mathbbm{1}_I \in P_\mathcal{M}$ such that $F(x^*) \leq O(\min\{ \frac{r}{c-2}, 1+\frac{\sigma}{r}\}) f(I) \leq  O(\sqrt{\sigma}) f(I)$.
\end{theorem}

Combining the continuous greedy methods with  this rounding procedure we obtain the following result which improves upon the previous best bound of $\Omega(1/\sigma^2)$.
\begin{theorem}
\label{thm:approx-sigma-diversity}
The problem of maximizing a $\sigma$-semi-metric diversity function over a matroid admits a $\Omega(1/\sigma^{3/2})$-approximation.
For uniform matroids this becomes a $\Omega(1/\sigma)$-approximation.
\end{theorem}

Theorem~\ref{thm:integrality-gaps} is obtained by two different rounding algorithms. One is based on modifying the approximate integer decomposition property \cite{chekuri2009approximate} to work for quadratic programs;  the second one adapts the swap rounding algorithm developed for submodular functions \cite{calinescu2011maximizing}. 
We discuss the first result here. For details regarding the second method, see Appendix~\ref{focsapp:rounding}. We remark that while our rounding results are inspired by previous techniques used for submodular maximization, the analysis requires several new insights to make it work for quadratic functions, since these are not convex in the $\bm{e}_i-\bm{e}_j$ directions.

\begin{theorem}
\label{thm:IGap}
Let $F$ be a non-negative, quadratic multilinear polynomial and $\mathcal{M}$ be a matroid with rank $r$ and minimum circuit size $c \geq 3$.
If $x^* \in P_{\mathcal{M}}$, then there is an independent set $I$  of $\mathcal{M}$ such that
$
(3+\frac{2r}{c-2})F(\mathbbm{1}_I) \geq F(x^*)
$.
\end{theorem}

We actually prove the following  decomposition result which implies Theorem~\ref{thm:IGap}.
For $x^* \in P_{\mathcal{M}}$, we define the  {\em coverage}  of a pair $u,v$  to be the quantity $x^*(u)x^*(v)$. Let $Cov \in \R^{{n \choose 2}}$ be the vector with entries $Cov(u,v)=x^*(u)x^*(v)$. As $F$ is quadratic it is linear in these coverage values and the vector $x^*$:  $F(x^*)=\sum_{u \neq v} (\frac{A(u,v)}{2}) Cov(u,v)+ \sum_v b(v) x^*(v)$.
For a set $X$ we say its {\em coverage set} is $cov(X)=\{ \{u,v\}:  u,v \in X, u \neq v\}$.
A {\em quadratic coverage} of $x^*$ is a collection
$\mathcal{C}=\{\mathbbm{1}_{I_i},\mu_i\}$ of weighted independent sets  with  properties (1) for each $u \neq v$,  $\sum_{i: \{u,v\} \subseteq cov(I_i)} \mu_i \geq Cov(u,v)$, and (2) for each $v$, $\sum_{i:I_i \ni v} \mu_i \geq x^*(v)$. 
Recall that $A,b\geq 0$.
It follows that
$\sum_i \mu_i F(\mathbbm{1}_{I_i}) \geq F(x^*)$ and hence  if
the {\em size} $\sum_i \mu_i \leq K$, then some $I_i$ satisfies $F(\mathbbm{1}_{I_i}) \geq \frac{F(x^*)}{K}$. 
This bound depends on the fact that entries of $A$ are non-negative. 
By condition (1) of quadratic coverages, we have $\sum_i \mu_i \mathbbm{1}_{cov(I_i)} \geq Cov$ and by condition (2), $\sum_i \mu_i \mathbbm{1}_{I_i}\geq x^*$. Therefore, for such a collection we have $\sum_i \mu_i F(\mathbbm{1}_{I_i}) \geq F(x^*)$. This reasoning shows that to deduce Theorem~\ref{thm:IGap}, it suffices to find a quadratic coverage with $\sum_i \mu_i \leq (3+\frac{2r}{c-2})$.

\begin{theorem}
\label{thm:quadcovround}
Let $F(x)=\frac{1}{2}x^T A x + b^T x$ be a non-negative, quadratic multilinear polynomial and $\mathcal{M}$ be a matroid with rank $r=r([n])$ and minimum circuit size $c \geq 3$.
If $x^* \in P_{\mathcal{M}}$, then it has a quadratic coverage
of size at most $3+\frac{2r}{c-2}$.
\end{theorem}
\begin{proof}
We start with an arbitrary representation of $x^*$ as a convex combination of independent sets: $\sum_i \lambda_i \mathbbm{1}_{B_i}$.

First note that $Cov(u,v)=(\sum_{B_i \ni u} \lambda_i) (\sum_{B_j \ni v} \lambda_j)=\sum_{(i,j): B_i \ni u, B_j \ni v} \lambda_i \lambda_j$.
Hence  an  ordered pair $(B_i,B_j)$ contributes $\lambda_i \lambda_j$ to $Cov(u,v)$  if
$u \in B_i,v \in B_j$.  This implies that if $B_i=B_j$, then this contributes exactly $\lambda_i^2$ for every $u,v \in B_i$.
If $B_i\neq B_j$, then the unordered pair $\{B_i,B_j\}$ contributes to coverages as follows.
It contributes $2 \lambda_i \lambda_j$ for every $u,v \in B_i \cap B_j$
and $\lambda_i \lambda_j$ for each $uv \in \delta(B_i-B_j,B_j-B_i,B_i \cap B_j)$.
Here for disjoint node sets $X_1,X_2, \ldots , X_p$ we define $\delta(X_1,X_2, \ldots ,X_p)$ to be the set of edges which have endpoints in distinct
sets from the $X_i$'s.
Hence we can express the coverage vector $Cov$ for $x^*$ in $\R^{{n \choose 2}}$ as: 
\begin{equation}
    \label{eqn:cov-vector}
\sum_i \lambda^2_i \cdot \mathbbm{1}_{cov(B_i)} + \sum_{i < j} \lambda_i \lambda_j  \cdot (2 \cdot \mathbbm{1}_{cov(B_i \cap B_j)} + \mathbbm{1}_{\delta(B_i-B_j,B_j-B_i,B_i \cap B_j)}).
\end{equation}

We now define a quadratic coverage, that is, a weighted collection of independent sets 
satisfying conditions (1) and (2). 
In particular, for  each
$i \leq j$ we  define a family of independent sets $\mathcal{I}^{i,j}$ which will take care of all coverages associated with terms $\lambda_i \lambda_j$ in (\ref{eqn:cov-vector}).
  In the case where $i=j$, this is easy. We just include the set
$B_i$ with weight $\mu_i=\lambda_i^2$. Now consider
the case where $i < j$ which is trickier.
  For each set $I$ in this family, we always associate the weight $\mu_I=\lambda_i \lambda_j$ and so this amounts to finding a family which satisfies
\begin{equation}
    \label{eqn:cov-vector2}
\sum_{I \in \mathcal{I}^{i,j}} \mathbbm{1}_{cov(I)} \geq 2 \cdot \mathbbm{1}_{cov(B_i \cap B_j)} + \mathbbm{1}_{\delta(B_i-B_j,B_j-B_i,B_i \cap B_j)}.
\end{equation}

\noindent
We return to this construction later but we note that condition (2) will follow easily
as long as we guarantee that for each $v,i$ and $j \neq i$, if $B_i \ni v$, then the family $\mathcal{I}^{i,j}$ includes at least one set $I$ which contains $v$.
Since we have $\mu_I=\lambda_i \lambda_j$ for any such $I$, we derive the 
desired inequality (2):
 $\sum_{I \ni v} \mu_I \geq \sum_{B_i \ni v}  (\sum_j \lambda_i \lambda_j) =  \sum_{B_i \ni v} \lambda_i = x^*(v)$. 

If we can achieve this construction so   that  $| \mathcal{I}^{i,j} | \leq K$ for each $i,j$, then
 we have a quadratic coverage whose size is 
 $\sum_{i} \mu_i + \sum_{i < j} \sum_{I \in \mathcal{I}^{i,j}} \mu_I = \sum_{i} \lambda^2_i + \sum_{i < j} \lambda_i \lambda_j |\mathcal{I}^{i,j}| \leq \sum_{i} \lambda^2_i + \sum_{i < j} \lambda_i \lambda_j K \leq 1 + K/2$. 
The last inequality follows since the $\lambda_i$ are a convex combination.

We now define $\mathcal{I}^{i,j}$ for a fixed pair $i,j$ and show how to find the desired  independent sets $\mathcal{I}^{i,j} = \{ I^{i,j}_k: k=1,2, \ldots ,K \}$, where $K$ is defined later.
First, if $|B_i \cap B_j| \geq 1$, then we include the sets $B_i,B_j$. This takes care of the double-coverage of pairs in $B_i \cap B_j$
as well as any pairs $u,v$ with $u \in B_i \cap B_j$ and $v \in B_i \Delta B_j$.  Let $S_{ij}=B_i \setminus B_j$ and $S_{ji}=B_j \setminus B_i$.
Note that the excess coverage from these sets $B_i,B_j$  is to contribute an extra
$\lambda_i \lambda_j$ to each pair in  $cov(S_{ij}) \cup cov(S_{ji})$.
It now remains
to cover the edges in $\delta(S_{ij},S_{ji})$. 

 Let $t=\lfloor (c-1)/2 \rfloor$ and $m=|B_i \cap B_j|\geq 0$. 
 Decompose $B_j \setminus B_i$ into 
$\ell=\lceil (r-m)/t \rceil$
disjoint independent sets
by ripping out sets of size $t$ greedily, possibly the last being smaller than $t$. Call these
$C_1,C_2, \ldots ,C_{\ell}$.
For each $k \leq \ell$, we extend $C_k$  to an independent set $R^{i,j}_k$ in $B_i \Delta B_j$ only adding elements from $B_i \setminus B_j$.
Hence this set will have used all elements of $B_i$ except
a subset, call it $Z_k$,   of size at most $t$. Let $C^{i,j}_k = Z_k \cup C_k$
and note that $|C^{i,j}_k| \leq 2t \leq c-1$ and hence it is also independent.
 We now examine the pairs covered by $C^{i,j}_k,R^{i,j}_k$.
 Let  $u \in C_k, v \in B_i \setminus B_j$, then either  $u,v$ is covered by $R^{i,j}_k$, or $v \in Z_k$ in which case it is  covered by  $C^{i,j}_k$.

Finally, we count the number of sets for a given family. There are two cases depending on whether $B_i \cap B_j = \emptyset$ or not.
If the intersection is empty, then we just build $2 \lceil \frac{r}{t} \rceil$. Since $t \geq \frac{c-2}{2}$, this is at most
$2 \cdot (1+\frac{2r}{c-2})$. In the other case we have $m \geq 1$,
and we add the sets $B_i,B_j$ up front and then we add $2 \lceil \frac{r-m}{t} \rceil$ more sets. 
Hence the overall number of sets in this case is at most $2+2\cdot (\frac{2r}{c-2} - \frac{2}{c-2}+1)$.

It follows that $K \leq 2 \cdot (2+\frac{2r}{c-2})$, and thus we have a quadratic coverage of size at most $1 + \frac{K}{2} \leq 3 + \frac{2r}{c-2} $, as we wanted to show.
\end{proof}

\section{Hardness}
\label{sec:hardness}

It is shown that it is hard to approximate the maximum of a metric diversity function subject to a cardinality constraint within a factor better than $\frac{1}{2}$~\cite{BhaskaraGMS16,BorodinLY12}. We generalize this hardness result to $\sigma$-semi-metric diversity functions. The following result shows that our approximation factor for maximizing a $\sigma$-semi-metric diversity function, subject to a uniform matroid (Theorem~\ref{thm:approx-sigma-diversity}) is asymptotically tight. For the proof of the following theorem, see Appendix~\ref{focsapp:hardness}. Let $\theta :=  n^{1/(\log \log n)^c}$ where
$c$ is a suitably chosen universal constant independent of $n$.

\begin{theorem}
\label{thm:approxhardness}
Assuming the exponential time hypothesis (ETH): (1)  There is no polytime $4/\theta$-approximation algorithm for maximizing  $\sigma$-semi-metric diversity functions subject to a cardinality constraint, and (2)
for any fixed $\sigma \geq 1$ and $\epsilon >0$, there is no polytime algorithm which  approximates the maximum of a $\sigma$-semi-metric diversity function subject to a cardinality constraint within a factor of $2\sigma-\epsilon$.
\end{theorem}

Combining Theorem~\ref{thm:approxhardness} and the $O(1)$ rounding for multilinear quadratics subject to a uniform matroid (Theorem \ref{thm:IGap}), gives the following result which states that the approximation bound given in Theorem~\ref{thm:thirdderiv}, for the functions with a non-positive third-order partial derivatives, is asymptotically tight. 

\begin{corollary}
\label{cor:hardness-dks}
Let $A$ be a matrix corresponding to a $\sigma$-semi-metric distance function. Then, assuming ETH, it is hard to approximate the  continuous problem $\max x^T A x: ||x||_1 \leq k$ within a factor of $o(\sigma)$. Moreover this implies that the analysis of the jump-start continuous greedy algorithm in Theorem~\ref{thm:thirdderiv} is asymptotically tight.
\end{corollary}

This result is conditioned on hardness of densest subgraph which
has been established under ETH \cite{manurangsi2017almost} - see Appendix~\ref{focsapp:hardness}.  
First, since the term $O(\frac{r}{c-2})$ in Theorem~\ref{thm:integrality-gaps}  does not depend on $\sigma$, it yields an $O(1)$ rounding gap for cardinality constraints (since $\frac{r}{c-2} \approx 1$).
In addition, given that the multilinear extension of the densest subgraph objective is of the form $x^T A x$, the approximability  of densest subgraph 
is within a constant factor of  its continuous relaxation.

The following result asserts that our rounding algorithm is also asymptotically tight. The proof is included in Appendix~\ref{focsapp:hardness}.

\begin{theorem}
\label{thm:intgaplower}
Let $k,t\in\mathbb{N}$ with $1 \leq t \leq k$. There exists a $\sigma$-semi-metric diversity function with multilinear extension $F$, and a matroid $\mathcal{M}=([2k],\I)$ with rank $r=k+t-1$ and minimum circuit size $c = 2t$, where the integrality gap of $F(x)$ over the matroid polytope $P_\mathcal{M}$ is $\Omega(\min \{ \frac{r}{c-2},\frac{\sigma}{r}\})$.
\end{theorem}

\section{Conclusion}

There are a number of  directions which need exploring. The most immediate are (i)  extending the continuous greedy algorithm  to non-monotone  $\sigma$-smooth functions, (ii)  develop rounding methods (such as contention resolution) for one-sided smooth functions over more general polytopes. We believe there should be further interesting applications for the one-sided smoothness model introduced in this work.

\section{Acknowledgements}

This article benefitted greatly from previous anonymous reviews.
We are indebted to those reviewers as well as to Chandra Chekuri, Anupam Gupta and Nick Harvey who also provided  invaluable feedback.
The third author gratefully acknowledges the support from an  NSERC
Discovery Grant 109840 without which this work would not be possible.

\bibliographystyle{plain}
\bibliography{references_diversity}

\appendix
\section{Appendix: Semi-metric diversity and OSS}
\label{focsapp:semimetricOSS}

In this section, we establish the smoothness parameter associated with several of the discrete quadratic  functions discussed.
  In other words, we bound the approximate triangle inequality for their associated distance functions.

\begin{definition}
Let $d:[n]\times [n]\rightarrow \mathbb{R}_{\geq 0}$ be a distance function with the corresponding distance matrix $D\in \mathbb{R}^{n\times n}_{\geq 0}$ where $D_{a,b}=d(a,b)$. We say $d$ is a negative-type distance if for any $x\in \mathbb{R}^n$ with $||x||_1=0$ we have $x^T Dx\leq 0$.
\end{definition}

\begin{proposition}
\label{prop:negtypeSemimetric}
Any negative-type distance $d:[n]\times [n]\rightarrow \mathbb{R}_{\geq 0}$ is $2$-semi-metric.
\end{proposition}
\begin{proof}
Let $x=0.5e_a+0.5e_b-e_c$. We know
\begin{align*}
x^T Dx & = 0.5d(a,b)-d(a,c)-d(b,c)   \leq 0.
\end{align*}
Therefore $d(a,b)\leq 2d(a,c)+2d(b,c)$ and $d$ is $2$-semi metric.
\end{proof}

Jensen-Shannon Divergence is a function which measures dissimilarity between probability distributions.
It is well-known that if $d$ is a JS measure, then  $\sqrt{d}$ is a metric. Hence JS distances form a $2$-semi-metric by the following result.

\begin{proposition}
\label{prop:sqrtdist}
Let $d:[n]\times [n]\rightarrow \mathbb{R}_{\geq 0}$ be a distance function such that $\sqrt{d(\cdot,\cdot)}$ is a metric. Then $d(\cdot,\cdot)$ is a  $2$-semi-metric.
\end{proposition}
\begin{proof}
By definition, we have
\[
\sqrt{d(i,j)}\leq \sqrt{d(i,k)}+\sqrt{d(j,k)}.
\]
Therefore,
\[
d(i,j) \leq d(i,k) + d(j,k) + 2\sqrt{d(i,k) d(j,k)}.
\]
We also know that
\[
d(i,k) + d(j,k) - 2\sqrt{d(i,k) d(j,k)} = (\sqrt{d(i,k)}-\sqrt{d(j,k)})^2 \geq 0.
\]
Hence,
\[
d(i,j) \leq 2(d(i,k) + d(j,k)).
\]
\end{proof}

\old{
\begin{lemma}
\label{lemma:probabilisticVersion}
Let $f$ be a non-negative, monotone set function and $F$ be its multi-linear function. Let $x\in [0,1]^n$ and $\sigma\geq 0$. If for any $i,j\in [n]$ we have
\[
\E_{R \sim x}[|R|]\cdot\E_{R \sim x}[A_{ij}(R)]\leq \sigma\cdot(\E_{R \sim x}[B_i(R)]+\E_{R \sim x}[B_j(R)]),
\]
or equivalently,
\[
||x||_1\nabla_{ij}^2 F(x) \leq \sigma (\nabla_i F(x)+ \nabla_j F(x)),
\]
then $F$ is one-sided $\sigma$-smooth at $x$.
\end{lemma}
\begin{proof}
We have
\begin{align*}
u^T \nabla^2 F(x) u & = \sum_{i=1}^n \sum_{j=1}^n u_i u_j \nabla_{ij}^2 F(x) \leq \frac{\sigma}{||x||_1} \sum_{i=1}^n \sum_{j=1}^n u_i u_j (\nabla_i F(x)+\nabla_j F(x))
\\ & = \frac{\sigma}{||x||_1}(\sum_{i=1}^n \sum_{j=1}^n u_i u_j \nabla_i F(x)+\sum_{i=1}^n \sum_{j=1}^n u_i u_j \nabla_j F(x))
\\ & = \frac{\sigma}{||x||_1}(\sum_{i=1}^n u_i \nabla_i F(x) (\sum_{j=1}^n u_j) +\sum_{i=1}^n u_i (\sum_{j=1}^n u_j \nabla_j F(x)))
\\ & = \frac{\sigma}{||x||_1}(||u||_1 \sum_{i=1}^n u_i \nabla_i F(x) + ||u||_1 \sum_{j=1}^n u_j \nabla_j F(x))
\\ & = 2\sigma \left(\frac{||u||_1}{||x||_1} \right) ( u^T \nabla F(x)).
\end{align*}
\end{proof}

\begin{proposition}
\label{prop:ossdiscretequad}
Let $A\in \mathbb{R}^{n\times n}$ be a symmetric, $0$-diagonal matrix. Let $b\in \mathbb{R}^n$ and $b\geq 0$. Then $F(x)=\frac{1}{2}x^T A x + b^T x$ is one-sided $\sigma$-smooth if $A$ is $\sigma$-semi-metric.
\end{proposition}
\begin{proof}
Note that $\nabla^2 F(x) = A$ and $\nabla F(x) = Ax+b$. Therefore,
\begin{align*}
\sigma(\nabla_i F(x)+\nabla_j F(x)) & \geq \sigma(\sum_{k=1}^n A(i,k)x_k + \sum_{k=1}^n A(j,k)x_k) = \sum_{k=1}^n \sigma(A(i,k)+A(j,k))x_k \\ & \geq \sum_{k=1}^n A(i,j)x_k = ||x||_1 A(i,j) = ||x||_1 \nabla^2_{ij}F(x),
\end{align*}
where the first inequality follows from $b\geq0$ and the last inequality holds because $A$ is $\sigma$-semi-metric. Now by Lemma~\ref{lemma:probabilisticVersion}, we conclude that $F$ is one-sided $\sigma$-smooth.
\end{proof}
}

\begin{lemma}
\label{lemma:probabilisticVersion}
Let  $F \in {\bf C}^2$, $x\in [0,1]^n$ and $\sigma\geq 0$. If for any $i,j\in [n]$ we have
\[
||x||_1\nabla_{ij}^2 F(x) \leq \sigma (\nabla_i F(x)+ \nabla_j F(x)),
\]
then $F$ is one-sided $\sigma$-smooth at $x$.
\end{lemma}
\begin{proof}
We have
\begin{align*}
u^T \nabla^2 F(x) u & = \sum_{i=1}^n \sum_{j=1}^n u_i u_j \nabla_{ij}^2 F(x) \leq \frac{\sigma}{||x||_1} \sum_{i=1}^n \sum_{j=1}^n u_i u_j (\nabla_i F(x)+\nabla_j F(x))
\\ & = \frac{\sigma}{||x||_1}(\sum_{i=1}^n \sum_{j=1}^n u_i u_j \nabla_i F(x)+\sum_{i=1}^n \sum_{j=1}^n u_i u_j \nabla_j F(x))
\\ & = \frac{\sigma}{||x||_1}(\sum_{i=1}^n u_i \nabla_i F(x) (\sum_{j=1}^n u_j) +\sum_{i=1}^n u_i (\sum_{j=1}^n u_j \nabla_j F(x)))
\\ & = \frac{\sigma}{||x||_1}(||u||_1 \sum_{i=1}^n u_i \nabla_i F(x) + ||u||_1 \sum_{j=1}^n u_j \nabla_j F(x))
\\ & = 2\sigma \left(\frac{||u||_1}{||x||_1} \right) ( u^T \nabla F(x)).
\end{align*}
\end{proof}

We have defined a symmetric matrix $A$ to be a $\sigma$-semi-metric (see Section~\ref{sec:intro}) if
 $A_{ik} \leq \sigma (A_{ij} + A_{jk})$ for all $i,j,k$.
 Our main applications are to multilinear extensions where
 $A$ is non-negative and  has  zero diagonal. However, the following result applies in the more general setting.

 \begin{proposition}
\label{prop:ossdiscretequad}
Let $A\in \mathbb{R}^{n\times n}$ be a non-negative symmetric matrix.
  Let $b\in \mathbb{R}^n$ and $b\geq 0$. Then $F(x)=\frac{1}{2}x^T A x + b^T x$ is one-sided $\sigma$-smooth if $A$ is a $\sigma$-semi-metric.
\end{proposition}
\begin{proof}
Note that $\nabla^2 F(x) = A$ and $\nabla F(x) = Ax+b$. Therefore,
\begin{align*}
\sigma(\nabla_i F(x)+\nabla_j F(x)) & \geq \sigma(\sum_{k=1}^n A(i,k)x_k + \sum_{k=1}^n A(j,k)x_k) = \sum_{k=1}^n \sigma(A(i,k)+A(j,k))x_k \\ & \geq \sum_{k=1}^n A(i,j)x_k = ||x||_1 A(i,j) = ||x||_1 \nabla^2_{ij}F(x),
\end{align*}
where the first inequality follows from $b\geq0$ and the last inequality holds because $A$ is $\sigma$-semi-metric. Now by Lemma~\ref{lemma:probabilisticVersion}, we conclude that $F$ is one-sided $\sigma$-smooth.
\end{proof}

\section{Appendix: Jump-Start Continuous Greedy}
\label{focsapp:contgreedy}

\begin{proposition}
\label{prop:bestc}
For any  $\sigma > 0$ the best approximation guarantee in Theorem~\ref{thm:greedyBound} is attained at
\[
\alpha = \frac{-(2\sigma+1) + \sqrt{4\sigma^2+12\sigma+1}}{2}.
\]
\end{proposition}
\begin{proof}
We need to find the maximizer of $g(\alpha)=(1-\alpha)(\frac{\alpha}{\alpha+1})^{2\sigma}$ where $\alpha\in[0,1)$. Hence, we solve $g'(\alpha)=0$.
\begin{align*}
& g'(\alpha) = \frac{2\sigma \alpha^{2\sigma-1} (\alpha+1)^{2\sigma} - (2\sigma+1)\alpha^{2\sigma} (\alpha+1)^{2\sigma} - 2\sigma (\alpha+1)^{2\sigma-1} \alpha^{2\sigma} + 2\sigma(\alpha+1)^{2\sigma-1} \alpha^{2\sigma+1}}{(\alpha+1)^{4\sigma}} = 0 \\ & \Rightarrow
2\sigma \alpha^{2\sigma-1} (\alpha+1)^{2\sigma-1} - 2\sigma \alpha^{2\sigma} (\alpha+1)^{2\sigma-1} = \alpha^{2\sigma} (\alpha+1)^{2\sigma} \\ & \Rightarrow
2\sigma \alpha^{2\sigma-1}(\alpha+1)^{2\sigma-1} (1-\alpha) = \alpha^{2\sigma} (\alpha+1)^{2\sigma} \\ & \Rightarrow
2\sigma (1-\alpha) = \alpha(\alpha+1) \Rightarrow \alpha^2 + (1+2\sigma)\alpha - 2\sigma = 0 \Rightarrow \alpha = \frac{-(2\sigma+1) \pm \sqrt{4\sigma^2+12\sigma+1}}{2}
\end{align*}
The only solution in $[0,1)$ is $\frac{-(2\sigma+1) + \sqrt{4\sigma^2+12\sigma+1}}{2}$ and this yields the proposition.
\end{proof}

\begin{reptheorem}{thm:thirdderiv}
	Let $F:[0,1]^n \to \mathbb{R}_{\geq 0}$ be a monotone $\sigma$-OSS function with non-positive third order partial derivatives.  Let  $\alpha \in [0,1)$ and
	 $P$ be a polytime separable, downward-closed, polytope. If we run the jump-start continuous greedy process (Algorithm~\ref{alg:ctsgreedy}) then $x(1) \in P$ and
     $
     F(x(1)) \geq  [1-\exp{(-\frac{\alpha(1-\alpha)}{\alpha+\sigma})}]\cdot OPT
     $
where $OPT := \max \{F(x):x\in P\}$. In particular, taking $\alpha=1/2$ we get $
     F(x(1)) \geq  [1-\exp{(-\frac{1}{4\sigma+2})}]\cdot OPT$ and so
     $F(x(1)) \geq  \frac{1}{4\sigma+3} \cdot OPT$ (since $e^{x} \geq x + 1$ for $x<1$).
\end{reptheorem}
\begin{proof}
	For each $t\in [0,1]$ we have
	\vspace{-0.2cm}
	\begin{equation}
		x(t) = x(0) + (1-\alpha) \int_{0}^{t} v_{max}(x(\tau)) \, d\tau =
		\alpha v^* + (1-\alpha) \int_{0}^{t} v_{max}(x(\tau)) \, d\tau.
	\end{equation}
	Since $P$ is convex and $v^* \in P$, we have that $x(t) \in P$ as long as $y(t):= \int_{0}^{t} v_{max}(x(\tau)) \, d\tau \in P$.
    Given that each $v_{max}(x(\tau)) \in P$ and also $\vec{0} \in P$, it follows that $y(t)$ is a convex combination of points in $P$, and hence belongs to $P$.

   Let $x^* \in P$ be such that $F(x^*)=OPT$. Also let $x \in \{ x(t):0\leq t\leq 1\}$ and $u=(x^*-x) \vee 0$, i.e.,  $x^* \vee x = x + u$.
   By Taylor's Theorem and non-positivity of the third order derivatives of $F$ we have
    \begin{align*}
    F(x^* \vee x) & \leq F(x) + u^T \nabla F(x) + \frac{1}{2} u^T \nabla^2 F(x) u
    \leq F(x) + \Big(1+\frac{\sigma ||u||}{||x||} \Big) u^T \nabla F(x)
    \\ & \leq F(x) + \Big(1+\frac{\sigma}{\alpha}\Big) u^T \nabla F(x),
    \end{align*}
    \noindent
    where the second inequality follows from smoothness, and the third from the fact that $||x(t)|| \geq ||x(0)|| = \alpha ||v^*|| \geq \alpha ||u||$. Thus
    \begin{equation}
    \label{eqn:key1}
    u^T \nabla F(x) \geq
    \Big(\frac{\alpha}{\alpha+\sigma} \Big) \Big(F(x \vee x^*) - F(x) \Big) \geq \Big(\frac{\alpha}{\alpha+\sigma}\Big)  \Big(OPT - F(x) \Big),
 \end{equation}
    where  the last inequality follows from  monotonicity. We also have that
    $$
    v_{max} (x)\cdot \nabla F(x)\geq x^* \cdot \nabla F(x) \geq u \cdot \nabla F(x),
    $$
    where the first inequality follows by definition of $v_{max}$ and the fact that $x^* \in P$, and the second inequality from the fact that $x^* \geq u$ and $\nabla F \geq 0$.
    Combining this with (\ref{eqn:key1})  yields:
	\begin{equation}
	\label{eq1b}
		v_{max}(x)\cdot \nabla F(x) \geq  \Big(\frac{\alpha}{\alpha+\sigma}\Big)  \Big(OPT - F(x) \Big),
	\end{equation}

    \noindent
    for any $x \in \{ x(t):0\leq t\leq 1\}$. Let us denote $\rho = \alpha / (\alpha + \sigma)$.
	We can use the Chain Rule to get
	\begin{equation}
	\label{eq21}
		\frac{d}{dt} F(x(t)) = \nabla F(x(t)) \cdot x'(t) = \nabla F(x(t)) \cdot (1-\alpha) v_{max}(x(t)) \geq\rho (1-\alpha) \Big[OPT - F(x(t))\Big],
	\end{equation}
	where the last inequality follows from (\ref{eq1b}).
	
	We solve the above differential inequality by multiplying by $e^{\rho (1-\alpha) t}$.
	\begin{eqnarray*}
		\frac{d}{dt} [e^{\rho (1-\alpha)t}\cdot F(x(t)) ] & = & \rho (1-\alpha)e^{\rho (1-\alpha)t}\cdot F(x(t)) + e^{\rho (1-\alpha)t}\cdot \frac{d}{dt} F(x(t)) \\
		 & \geq &  \rho (1-\alpha)e^{\rho (1-\alpha)t}\cdot F(x(t)) + \rho \cdot e^{\rho (1-\alpha)t} (1-\alpha) [OPT - F(x(t))] \\
		 & = & \rho (1-\alpha)e^{\rho (1-\alpha)t}\cdot OPT.
	\end{eqnarray*}
	where the inequality follows from Equation (\ref{eq21}).
	
	Integrating the LHS and RHS of the above equation between $0$ and $t$ we get
	\begin{eqnarray*}	
		e^{\rho (1-\alpha)t}\cdot F(x(t)) - e^0 \cdot F(x(0)) & \geq & \rho (1-\alpha)OPT \int_{0}^{t} e^{\rho (1-\alpha)\tau} d\tau \\ &=&
		\rho (1-\alpha)OPT \cdot [\frac{e^{\rho (1-\alpha)t}}{\rho (1-\alpha)} - \frac{1}{\rho (1-\alpha)}]
		 = OPT \cdot [e^{\rho (1-\alpha)t}-1].
	\end{eqnarray*}
	Hence
	\begin{eqnarray*}	
	F(x(t))  & \geq & [1-\frac{1}{e^{\rho (1-\alpha)t}}] OPT + \frac{F(x(0))}{e^{\rho (1-\alpha)t}} \geq [1-\frac{1}{e^{\rho (1-\alpha)t}}] OPT,
	\end{eqnarray*}
	where the last inequality follows from the fact that $F$ is non-negative.
    Substituting $t=1$ and $\rho = \alpha / (\alpha + \sigma)$ gives the desired result.
\end{proof}

\section{Appendix: Discretization of the Continuous Greedy}
\label{focsapp:discretization}

We now discuss discretization of the continuous greedy process for one-sided smooth functions.

 If our goal is to find a polytime approximation algorithm, we need to establish two features. The first is an
 approximation bound; for this we  use our analysis of the continuous greedy process, Theorems~\ref{thm:greedyBound} and \ref{thm:thirdderiv}.  The second is some sort of smoothness assumption on the gradients of $F$. We consider several conditions for the latter depending on the context; the most straightforward is for multilinear \OSS functions.

To discretize the jump-start continuous greedy, we start at $x^0 = \alpha v^*$ (for $\alpha \in [0,1)$) and use the following update rule.
\[
x^{t+\delta} = x^t + \delta \cdot (1-\alpha) \cdot v_{max}(x^t),
\]
where $\delta$ is the step size, and $v_{max}(x) := \argmax_{v \in P} v^T \nabla F(x)$.

We always assume $\delta >0$ is chosen with $1/\delta$ integer, which is then clearly the number of iterations. Our main concern is to bound this by a polynomial in the input size. This is because we  primarily adopt the view that we have exact access to the function $F$ and its gradients. This  is the case for the ME of a diversity function (and in fact any quadratic function) which is our main application. For more general functions
we may not have access to the exact gradient and we should find an estimate by sampling from the function. In that case, we  need a probabilistic argument similar to the original argument of  Vondrak \cite{vondrak2008optimal}.

 There are two ingredients we need to analyze discretizations. One is an approximation bound for the continuous process itself. The second is a bound which guarantees that gradients do not decrease too suddenly.  We describe the discretization as a self-contained argument which  takes these two bounds ((\ref{eqn:mubound}) and (\ref{eqn:etabound})) as inputs.

	Let $F:[0,1]^n\rightarrow \mathbb{R}$ be a $\sigma$-OSS function and $P$ be a  downward-closed polytope. Denote $OPT := \max_{x \in P} F(x)$. We consider  generic lower bounds on the continuous greedy rate of improvement as a function of $\alpha$ and $\sigma$.
	For some $\mu = \mu(\alpha,\sigma) \in (0,1]$, we say an application satisfies a {\em $\mu$ bound }
	if for any $x \in P$ such that $x \geq \alpha v^*$ we have
	\begin{equation}
	\label{eqn:mubound}
 v_{max}(x) \cdot \nabla F(x) \geq 	
	\mu (OPT - F(x))
	\end{equation}

The following lemma encapsulates the two main bounds $\mu$ we use; these
are outcomes of the  proofs of Theorem~\ref{thm:greedyBound} and Theorem~\ref{thm:thirdderiv}.
\begin{lemma}
	\label{lem:3.1}
	Let $F:[0,1]^n\rightarrow \mathbb{R}$ be a $\sigma$-OSS function and $P$ be a  downward-closed polytope. Denote $OPT := \max_{x \in P} F(x)$. Then for any $x \in P$ such that $x \geq \alpha v^*$ we have
	\[
 v_{max}(x) \cdot \nabla F(x) \geq  \Big(\frac{\alpha}{\alpha+1} \Big)^{2\sigma} (OPT-F(x))
	\]

\noindent
If in addition $F$ has non-positive third derivatives, then we have
	\[
 v_{max}(x) \cdot \nabla F(x) \geq 	 \Big(\frac{\alpha}{\alpha+\sigma} \Big) (OPT-F(x))
	\]
\end{lemma}

 For
 $\eta \geq 0$ (possibly a function of inputs such as $n$), we say
 $F \in {\bf C}^2$ is {\em $\eta$-local at $x,u$} if
    \begin{equation}
    \label{eqn:etabound}
        u^T \nabla F(x+ \epsilon u) \geq (1-\eta \epsilon) u^T \nabla F(x)
    \end{equation}
    for all  $\epsilon \in [0,1]$  such that   $x+\epsilon u \in P$ and where $1-\eta \epsilon \in [0,1]$. The function is {\em $\eta$-local} if this holds
    for all such choices of $x,u$.
    The next result shows how one may obtain a polytime implementation
    of continuous greedy for functions with ``bounded locality''.
    As discussed later, in some  applications, functions may only be local for a subset of  $x,u$.

\begin{theorem}
\label{thm:discretize_eta}
    Let $F:[0,1]^n\rightarrow \mathbb{R} \in {\bf C}^2$ be a monotone, non-negative $\sigma$-OSS function and $P$ a polytime separable downward-closed polytope. Assume $F$ satisfies
    a $\mu = \mu(\alpha,\sigma)$ bound and is $\eta$-local.
    Then taking $\delta\leq \min \{\frac{1}{n \eta (1-\alpha)},\frac{1}{(1-\alpha)\mu}\}$,  discrete greedy   produces $x^{1}$ satisfying:
	\begin{equation}
	F(x^{1}) \geq  (1-\exp(-(1-\alpha) \mu) (1-o(1))OPT
	\end{equation}
\end{theorem}
\begin{proof}
	By definition of the algorithm we have $x^0 = \alpha v^*$ and $x^{t+\delta} = x^t + \delta \cdot (1-\alpha) \cdot v_{max}(x^t)$.
	Then by Taylor's Theorem  for some $\epsilon \in [0,1]$ we have
	\begin{align*}
	F(x^{t+\delta}) &= F( x^t + \delta (1-\alpha)  v_{max}(x^t) ) 	
	= F( x^t) +  \delta  (1-\alpha)  v_{max}(x^t) \cdot \nabla F(x^t + \epsilon \delta (1-\alpha)  v_{max}(x^t)) \\ &
	\geq F( x^t) +  (1-\eta \epsilon \delta (1-\alpha)) \delta (1-\alpha)  v_{max}(x^t) \cdot \nabla F(x^t) \\ &
	\geq F( x^t) +  (1-\eta  \delta (1-\alpha)) \delta (1-\alpha) \mu  \big(OPT - F(x^t) \big)
	\end{align*}
	where the first inequality follows from $\eta$-locality and the second inequality follows by the $\mu$ bound property, and $\epsilon \in [0,1]$.
	Now define $\tilde{OPT}= (1-\eta \delta (1-\alpha)) OPT$.
	We have
	\begin{align*}
	F(x^{t+\delta}) & \geq  F( x^t) +  \delta (1-\alpha) \mu \left(\tilde{OPT} - F(x^t)\right)
	\end{align*}

\noindent
because $F$ is non-negative and $(1-\eta  \delta (1-\alpha)) \in [0,1]$
by choice of $\delta$. Hence we have
\[
(1-\delta (1-\alpha)\mu)(\tilde{OPT}-F(x^t))\geq \tilde{OPT}-F(x^{t+\delta}).
\]
By induction, we have
\[
(1-\delta (1-\alpha)\mu)^{1/\delta}(\tilde{OPT}-F(x^0))\geq \tilde{OPT}-F(x^{1})
\]
Next, since $\delta \leq \frac{1}{(1-\alpha)\mu}$ , we have $(1-\delta(1-\alpha)\mu)^{1/\delta}\leq \exp(-(1-\alpha)\mu)$.  Therefore we have
\begin{align*}
F(x^1) & \geq  (1-\exp(-(1-\alpha)\mu)\tilde{OPT} \\ & = (1-\exp(-(1-\alpha)\mu)(1-\eta  \delta (1-\alpha)) OPT \\ & = (1-\exp(-(1-\alpha)\mu)(1-o(1)) OPT,
\end{align*}
where the first inequality holds because of non-negativity of $F$, and the last equality holds because $\delta \leq \frac{1}{n \eta (1-\alpha)}$.

\end{proof}

We discuss how one may apply this theorem to functions $F$
   with gradients that are $L$-Lipschitz (with respect to $\ell_2$ norm). 
   It follows that  $|u^T (\nabla F(x+ \epsilon u) -  \nabla F(x))| \leq ||u||_2 \cdot ||\nabla F(x+ \epsilon u) -  \nabla F(x)||_2 \leq \epsilon L  ||u||_2^2  \leq \epsilon n^2 L$. Define  $\eta=n^3 L$ and suppose that for some $x$ we have
   that $u=v_{max}(x)$ fails the condition for $\eta$-locality.  
   That is,  $u^T \nabla F(x+ \epsilon u) < (1- \epsilon n^3 L) u^T \nabla F(x)$
   and hence $|u^T (\nabla F(x+ \epsilon u) -  \nabla F(x))| > \epsilon n^3 L( u^T \nabla F(x))$. Together with the first inequality this yields:
      $u^T \nabla F(x) < \frac{1}{n}$. Hence if $F$ satisfies (\ref{eqn:mubound}) we have
      \[
      \frac{1}{n} > v_{max} \cdot \nabla F(x) \geq \mu (OPT-F(x)).
      \]
      
      \noindent
      It follows that $F(x) > OPT - \frac{1}{\mu n}$. Hence if we  follow the analysis in the proof of Theorem~\ref{thm:discretize_eta}, either we  achieve the claimed multiplicative bound, or we reach a point $x$ which is within a small additive constant of {\sc opt}.

We now apply discretization to our main applications. Note that in some cases, it is enough to have the locality condition on a subdomain of the function.
One may show that for non-negative monotone set functions,
their multilinear extensions are $n^2$-local on $\{x: x\leq \frac{\vec{1}(n-1)}{n}\}$. This yields the following result.

\begin{theorem}
\label{thm:discretizemultilinear}
Let $f$ be a non-negative monotone set function and $F$ be its multilinear extension such that $F$ is $\sigma$-OSS.
Let $P$ denote a polytime separable downward-closed polytope contained in $[0,1]^n$. Assume that $F$ satisfies some $\mu=\mu(\alpha,\sigma)$ bound.
Then the output of the discrete version of jump-start continuous greedy algorithm, with $\delta \leq \min\{\frac{1-\alpha}{n^3},\frac{1}{(1-\alpha)\mu}\}$,   satisfies
\[
F(x^1)\geq (1-\exp(-\frac{1}{2}(1-\alpha)\mu)(1-o(1)) OPT
\]
\end{theorem}
\begin{proof}
First of all, suppose $t\leq 1-\frac{1}{n}$. By Taylor's remainder theorem, for some $\epsilon\in[0,1]$, we have
\[
F(x^{t+\delta})=F(x^t +\delta(1-\alpha)v_{max}(x^t))=F(x^t)+\delta (1-\alpha) v_{max}(x^t)\cdot\nabla F(x^t +\epsilon\delta(1-\alpha)v_{max}(x^t))
\]
Now note that because $F$ is the multi-linear extension of $f$, we have
\[
\nabla_i F(x)=\sum_{R\subseteq [n]} (f(R+i)-f(R-i)) p_x(R),
\]
where $p_x(R)=\prod_{j\in R} x_j \prod_{j\notin R} (1-x_j)$ --- see \cite{vondrak2008optimal}.
Because $f$ is monotone, the term $f(R+i)-f(R-i)$ is non-negative for any $R$ and $i$. Also note that because $P \subseteq [0,1]$ is downward-closed we have $v_{max}(x^t)\geq 0$.

Since $x^0_i \leq \alpha$ and $t\leq 1-\frac{1}{n}$, for any $i$ we have $x^t_i\leq \alpha+(1-\alpha)(1-\frac{1}{n})=1-\frac{1-\alpha}{n}$. Let $y^t=x^t +\epsilon\delta(1-\alpha)v_{max}(x^t)$. Because $v_{max}(x^t)\in[0,1]^n$ and $\epsilon,\alpha\in [0,1]$, we also have $x^t_j+\delta\geq y^t_j \geq x^t_j$. Therefore
\begin{align*}
p_{y^t}(R) & =\prod_{j\in R} y^t_j \prod_{j\notin R} (1-y^t_j) \geq \prod_{j\in R} x^t_j \prod_{j\notin R} (1-x^t_j -\delta).
\end{align*}
Now note that because $x_j^t\leq 1-\frac{1-\alpha}{n}$, we have $-\frac{\delta}{1-x_j^t}\geq -\frac{n\delta}{1-\alpha}$.
Therefore we have $\frac{(1-x_j^t-\delta)}{1-x_j^t}\geq 1- \frac{n\delta}{1-\alpha}$. Hence, by Bernoulli's inequality
and choice of $\delta$, we have
\[p_{y^t}(R)\geq (1-\frac{n\delta}{1-\alpha})^n p_{x^t}(R) \geq (1-\frac{n^2\delta}{1-\alpha}) p_{x^t}(R).\]
Hence,
\[
\nabla_i F(y^t) = \nabla_i F(x^t +\epsilon\delta(1-\alpha)v_{max}(x^t)) \geq (1-\frac{n^2\delta}{1-\alpha}) \nabla_i F(x^t).
\]
Therefore, defining $\tilde{OPT} = (1-\frac{n^2\delta}{1-\alpha}) OPT$, we have
\begin{align*}
F(x^{t+\delta}) & \geq F(x^t) + (1-\frac{n^2\delta}{1-\alpha}) \delta (1-\alpha) v_{max}(x^t)\cdot\nabla F(x^t) \\ & \geq F(x^t) + (1-\frac{n^2\delta}{1-\alpha}) \delta (1-\alpha)\mu (OPT-F(x^t)) \\ & \geq F(x^t) + \delta (1-\alpha)\mu (\tilde{OPT}-F(x^t)),
\end{align*}
where the second inequality follows from assumption.  The last inequality holds because $F$ is non-negative and $(1-\frac{n^2\delta}{1-\alpha})\leq 1$. Hence we have
\[
(1-\delta (1-\alpha)\mu)(\tilde{OPT}-F(x^t))\geq \tilde{OPT}-F(x^{t+\delta}).
\]
Therefore by taking $t=1-1/n$ and using induction, we have
\[
(1-\delta (1-\alpha)\mu)^{\frac{1}{\delta}(1-\frac{1}{n})+1}(\tilde{OPT}-F(x^0))\geq \tilde{OPT}-F(x^{1-\frac{1}{n}+\delta})
\]
Note that for any $n \geq 2$, $\frac{1}{\delta}(1-\frac{1}{n})+1\geq \frac{1}{2\delta}$.
Also note that $F(x^0)\geq 0$. Therefore
\[
(1-\delta (1-\alpha)\mu)^{1/2\delta}\tilde{OPT}\geq \tilde{OPT}-F(x^{1-\frac{1}{n}+\delta})
\]
Note that $(1-\frac{2\delta}{2} (1-\alpha)\mu)^{1/2\delta}\leq \exp(-\frac{1}{2}(1-\alpha)\mu)$, as $\frac{(1-\alpha)\mu}{2} \leq \frac{1}{2\delta}$. Therefore we have
\begin{align*}
F(x^1) & \geq  F(x^{1-\frac{1}{n}+\delta}) \geq (1-\exp(-\frac{1}{2}(1-\alpha)\mu)\tilde{OPT} \\ & \textcolor{red}{=} (1-\exp(-\frac{1}{2}(1-\alpha)\mu)(1-\frac{n^2\delta}{1-\alpha}) OPT \\ & = (1-\exp(-\frac{1}{2}(1-\alpha)\mu)(1-o(1)) OPT,
\end{align*}
where the first inequality holds because of monotonicity, and the last equality holds because $\delta \leq \frac{1}{(1-\alpha)n^3}$.
\end{proof}
	
We may also show that  the discrete greedy algorithm achieves $1$-step convergence to the claimed bounds if in addition $F$ has an even stronger lower bound on its gradients. (A property effectively saying $F$ is $0$-local.) As we see, this property is satisfied for multilinear extensions of supermodular functions.


\begin{theorem}
\label{thm:discretizebetageneral}
    Let $F:[0,1]^n\rightarrow \mathbb{R}$ be a monotone, non-negative $\sigma$-OSS function and $P$ a polytime separable downward-closed polytope. Assume $F$ satisfies
    a $\mu = \mu(\alpha,\sigma)$ bound and in addition:
    \begin{equation}
    \label{eqn:betabound}
        u^T \nabla F(x+ \epsilon u) \geq \beta u^T \nabla F(x)
    \end{equation}
    for all $u,x \in P$ and $\epsilon \in [0,1]$ such that $x+\epsilon u \in P$. (We assume $\beta \in (0,1]$ and the larger the value of $\beta$ the better). Then $1$-step discrete continuous greedy computes $x^{1}$ satisfying:
	\begin{equation}
	F(x^{1}) \geq  (1-\exp(-\beta (1-\alpha) \mu) OPT
	\end{equation}
\end{theorem}
\begin{proof}
	By definition of the algorithm we have $x^0 = \alpha v^*$ and $x^{t+\delta} = x^t + \delta \cdot (1-\alpha) \cdot v_{max}(x^t)$.
	Then by Taylor's Theorem  for some $\epsilon \in [0,1]$ we have
	\begin{align*}
	F(x^{t+\delta}) &= F( x^t + \delta (1-\alpha)  v_{max}(x^t) ) 	
	= F( x^t) +  \delta  (1-\alpha)  v_{max}(x^t) \cdot \nabla F(x^t + \epsilon \delta (1-\alpha)  v_{max}(x^t)) \\ &
	\geq F( x^t) +  \beta \delta (1-\alpha)  v_{max}(x^t) \cdot \nabla F(x^t) \\ &
	\geq F( x^t) +  \beta \delta (1-\alpha) \mu \left(OPT - F(x^t) \right) \\ &
	= \left(1-\beta \delta (1-\alpha) \mu  \right) F( x^t) +  \beta \delta (1-\alpha)  \mu OPT.
	\end{align*}
	where the first inequality follows from lemma's assumption, and the second inequality follows by the $\mu$ bound property.
	For the $1$-step version we have $\delta=1$ and so
	\begin{equation*}
	OPT - F(x^{(1)}) \leq   \left(1-\beta  (1-\alpha)  \mu \right) \big(OPT - F(x^0) \big) \leq \left(1-\beta  (1-\alpha)  \mu \right) OPT.
	\end{equation*}
Thus
	\begin{equation*}
	F(x^1) \geq  (1- \left(1-\beta  (1-\alpha)  \mu \right))  OPT \geq (1-exp(-\beta  (1-\alpha)  \mu)) OPT,
	\end{equation*}
	where the last step uses the exponential inequality
	$1-x \leq e^{-x}$.
\end{proof}

\begin{remark}
Note that for $\beta =1$ the above approximation factor is equal to $[1-\exp{(-(1-\alpha)(\frac{\alpha}{\alpha+1})\strut\strut^{2\sigma})}]$, which matches the approximation obtained via the continuous greedy process, i.e., Theorem \ref{thm:greedyBound}.
\end{remark}
	
\begin{lemma}
A $\sigma$-\OSS function $F \in {\bf C}^2$ satisfies (\ref{eqn:betabound}) with $\beta=1$ if  $\nabla^2 F(x)$ is a copositive matrix
for any $x$.  In particular, the multilinear extension of a supermodular function has $\beta = 1$.
\end{lemma}
\begin{proof}
Using fundamental theorem of calculus, we have  $\nabla F(x+p)= \nabla F(x) + \int_0^1 \nabla^2 F(x + t p) p dt$. Now the first part of the lemma  follows with $p=\epsilon u$
and taking the inner product with $u$, since
$u^T \nabla^2 F(x + t \epsilon u) u \geq 0$ for any $u \geq 0$.
For the second part, let $F$ be the multilinear extension of a supermodular set function $f$.  Then $-F$ is a multilinear extension of the submodular set function $-f$.  Vondrak \cite{vondrak2008optimal} shows that the Hessian of $-F$ is always non-positive with $0$ diagonal. Thus
$\nabla^2 F(x) \geq 0$ and hence copositive.
\end{proof}

The version of discrete greedy for multilinear extensions of supermodular functions may appear too good in that it only requires one step.
It has two  intensive computational ingredients, however. First is to solve an LP to find  a starting iterate $x^0=\alpha v^*$. The second is to compute the gradient $\nabla F(x^0)$, which already requires $O(n^2)$ work.


\section{Appendix: Swap Rounding for multilinear quadratics}
\label{focsapp:rounding}

In this section, we analyze a modified version of the swap rounding algorithm (Algorithm~\ref{alg:swapRounding}) and we show that it finds an integral solution which is an $O(1+\frac{\sigma}{r})$-approximation of the initial fractional solution.

First we define the following notation. $A(S) = \frac{1}{2}\sum_{ i,j\in S} A(i,j)$ and $A(S,S') = \sum_{i\in S}\sum_{j\in S'} A(i,j)$ and $b(S) = \sum_{i\in S} b(i)$. With an abuse of notation, we show $A(\{i\},S)$ with $A(i,S)$. The following result provides a decomposition of the multilinear extension of a quadratic function based on the convex decomposition of a point to the bases of the matroid.
\begin{lemma}
\label{lem:2ndordermulti}
Let $f(S)=\sum_{i\in S} b(i)+\frac{1}{2}\sum_{i,j\in S} A(i,j)$ where $b:[n]\rightarrow \mathbb{R}_{\geq 0}$ and $A:[n]\times[n] \rightarrow \mathbb{R}_{\geq 0}$ is a symmetric matrix with $A(i,i)=0$ for all $i \in [n]$. Then the multilinear extension of $f$ is $F(x) = \frac{1}{2} x^T A x + x^T b$. Moreover, if $x=\sum_{k=1}^p \lambda_k \mathbbm{1}_{I_k} $ for some scalars $\lambda_k$'s and subsets $I_k \subseteq [n]$, then
\begin{align}
\label{eq:2ndordermulti}
F(x) = \sum_{k=1}^p \lambda_k b(I_k) + \sum_{k=1}^p \lambda_k^2 A(I_k) + \sum_{k=1}^{p-1}\sum_{\ell=k+1}^{p} \lambda_k \lambda_{\ell} A(I_k, I_\ell).
\end{align}
\end{lemma}
\begin{proof}
For the first part of the lemma note that
\begin{align*}
F(x) & = \sum_{S\subseteq [n]} f(S) \prod_{k\in S} x_k \prod_{k\in [n]\setminus S} (1-x_k) = \sum_{S\subseteq [n]} (b(S)+A(S)) \prod_{k\in S} x_k \prod_{k\in [n]\setminus S} (1-x_k)  \\ & = \sum_{S\subseteq [n]} (\sum_{i\in S}b(i)) \prod_{k\in S} x_k \prod_{k\in [n]\setminus S} (1-x_k)  + \sum_{S\subseteq [n]} (\frac{1}{2}\sum_{i,j\in S}A(i,j)) \prod_{k\in S} x_k \prod_{k\in [n]\setminus S} (1-x_k)  \\ & =
\sum_{i\in [n]} b(i) \sum_{\substack{S\subseteq [n]\\i\in S}}( \prod_{k\in S} x_k \prod_{k\in [n]\setminus S} (1-x_k) ) + \frac{1}{2}\sum_{i,j\in [n]} A(i,j) \sum_{\substack{S\subseteq [n]\\\{i,j\}\subseteq S}} (\prod_{k\in S} x_k \prod_{k\in [n]\setminus S} (1-x_k) ) \\ & = \sum_{i\in [n]} b(i) x_i \sum_{\substack{S\subseteq [n]-i}}( \prod_{k\in S} x_k \prod_{k\in [n]-i\setminus S} (1-x_k)) \\ & + \frac{1}{2}\sum_{i,j\in [n]} A(i,j) x_i x_j \sum_{\substack{S\subseteq [n]-i-j}} (\prod_{k\in S} x_k \prod_{k\in [n]-i-j\setminus S} (1-x_k)) \\ & =
\sum_{i\in [n]} b(i) x_i + \frac{1}{2}\sum_{i,j\in [n]} A(i,j) x_i x_j = x^T b + \frac{1}{2} x^T A x.
\end{align*}
To see the second part, observe that
\begin{equation*}
    b^T x = b^T (\sum_k \lambda_k \mathbbm{1}_{I_k}) = \sum_k \lambda_k (b^T \mathbbm{1}_{I_k}) = \sum_k \lambda_k b(I_k),
\end{equation*}
and
\begin{align*}
    x^T A x & = (\sum_{k=1}^p \lambda_k \mathbbm{1}_{I_k})^T A (\sum_{\ell=1}^p \lambda_\ell \mathbbm{1}_{I_\ell})= \sum_{k,\ell=1}^p \lambda_k \lambda_\ell \mathbbm{1}_{I_k}^T A  \mathbbm{1}_{I_\ell} = \sum_{k,\ell=1}^p \lambda_k \lambda_\ell A(I_k, I_\ell) \\ &  =  \sum_{k=1}^p \lambda_k^2 A(I_k,I_k) + 2 \sum_{k<\ell} \lambda_k \lambda_\ell A(I_k,I_\ell) = 2 \sum_{k=1}^p \lambda_k^2 A(I_k) + 2\sum_{k=1}^{p-1}\sum_{\ell=k+1}^{p} \lambda_k \lambda_\ell A(I_k,I_\ell).
\end{align*}

\end{proof}

\begin{lemma}
\label{lemma:swaplem}
Let $\mathcal{M}=([n],\mathcal{I})$ be a matroid and $P$ be its corresponding base polytope. Let $F(z)=\frac{1}{2}z^T A z + z^T b$ where $A, b\geq 0$ and $A$ is a symmetric matrix such that its diagonal is zero. Let $f(S)=F(\mathbbm{1}_S)$ for any $S\subseteq [n]$. Let $x=\sum_{i=1}^p \lambda_i \mathbbm{1}_{I_i}  \in P$ where $I_i$'s are bases of the matroid, $\sum_{i=1}^p \lambda_i = 1$, and $\lambda_i\geq 0$, for $i=1,\ldots,p$. Let $(I',M)$ be the output of \texttt{\textsc{MergeBases}} (defined in Algorithm~\ref{alg:swapRounding}) on $(I_1,\ldots,I_p)$ and $(\lambda_1,\ldots,\lambda_p)$. Let $y = (\lambda_1+\lambda_2)\mathbbm{1}_{I'} + \sum_{i=3}^p \lambda_i \mathbbm{1}_{I_i}$. Then $F(x)\leq F(y)+\lambda_1 \lambda_2 \sum_{(i,j)\in M} A(i,j)$.
\end{lemma}
\begin{proof}
Let $I_1^0 = I_1$ and $I_2^0=I_2$ (the original inputs of the function). Let $I_1^m$ and $I_2^m$ be the resulting $I_1$ and $I_2$ after the $m$-th iteration of the while loop. Let $x_m=\lambda_1 \mathbbm{1}_{I_1^m} + \lambda_2 \mathbbm{1}_{I_2^m} + \sum_{k=3}^p \lambda_k \mathbbm{1}_{I_k}$. Let $i_m,j_m$ be the elements we pick at the $m$-th iteration of the loop. We show that $F(x_{m-1})\leq F(x_{m})+\lambda_1 \lambda_2 A(i_m,j_m)$ and this yields the desired result using a simple recursion argument. Without loss of generality, we assume
\begin{align}
\label{eq:swaplem}
 \nonumber & b(i_{m})+\lambda_1 A(i_m, I^{m-1}_1-i_m)+\lambda_2 A(i_m, I^{m-1}_2-j_m) + \sum_{k=3}^{p} \lambda_k A(i_m, I_k) \\ & \geq b(j_{m})+ \lambda_1 A(j_m,I^{m-1}_1-i_m) + \lambda_2 A(j_m,I^{m-1}_2-j_m) + \sum_{k=3}^{p} \lambda_k A(j_m, I_k)
\end{align}
We have
\begin{align*}
    F & (x_{m-1})  = \lambda_1 b(I^{m-1}_1) + \lambda_2 b(I^{m-1}_2) + \sum_{k=3}^p \lambda_k b(I_k) + \lambda_1^2 A(I^{m-1}_1)+\lambda_2^2 A(I^{m-1}_2) + \sum_{k=3}^p \lambda_k^2 A(I_k) \\ & +
    \lambda_1 \lambda_2 A(I_1^{m-1},I_2^{m-1}) + \lambda_1 \sum_{k=3}^p \lambda_k A(I_1^{m-1},I_k) + \lambda_2 \sum_{k=3}^p \lambda_k A(I_2^{m-1},I_k) + \sum_{k=3}^{p-1} \sum_{k'=k+1}^{p} \lambda_k \lambda_{k'} A(I_k,I_{k'}) \\ & = \lambda_1 b(I^{m-1}_1) + \lambda_2 b(I^{m-1}_2 - j_m) + \sum_{k=3}^p \lambda_k b(I_k) + \lambda_1^2 A(I^{m-1}_1)+\lambda_2^2 A(I^{m-1}_2 - j_m) + \sum_{k=3}^p \lambda_k^2 A(I_k) \\ & +
    \lambda_1 \lambda_2 A(I_1^{m-1},I_2^{m-1} - j_m) + \lambda_1 \sum_{k=3}^p \lambda_k A(I_1^{m-1},I_k) + \lambda_2 \sum_{k=3}^p \lambda_k A(I_2^{m-1}-j_m,I_k) \\ & + \sum_{k=3}^{p-1} \sum_{k'=k+1}^{p} \lambda_k \lambda_{k'} A(I_k,I_{k'}) + \lambda_2 b(j_m) + \lambda_2^2 A(j_m,I_2^{m-1}-j_m) + \lambda_1 \lambda_2 A(j_m, I_1^{m-1} - i_m) \\ & + \lambda_2 \sum_{k=3}^p \lambda_k A(j_m,I_k) + \lambda_1 \lambda_2 A(i_m,j_m) \\ & \leq
    \lambda_1 b(I^{m-1}_1) + \lambda_2 b(I^{m-1}_2 - j_m) + \sum_{k=3}^p \lambda_k b(I_k) + \lambda_1^2 A(I^{m-1}_1)+\lambda_2^2 A(I^{m-1}_2 - j_m) + \sum_{k=3}^p \lambda_k^2 A(I_k) \\ & +
    \lambda_1 \lambda_2 A(I_1^{m-1},I_2^{m-1} - j_m) + \lambda_1 \sum_{k=3}^p \lambda_k A(I_1^{m-1},I_k) + \lambda_2 \sum_{k=3}^p \lambda_k A(I_2^{m-1}-j_m,I_k) \\ & + \sum_{k=3}^{p-1} \sum_{k'=k+1}^{p} \lambda_k \lambda_{k'} A(I_k,I_{k'}) + \lambda_2 b(i_m) + \lambda_2^2 A(i_m,I_2^{m-1}-j_m) + \lambda_1 \lambda_2 A(i_m, I_1^{m-1} - i_m) \\ & + \lambda_2 \sum_{k=3}^p \lambda_k A(i_m,I_k) + \lambda_1 \lambda_2 A(i_m,j_m) \\ & =
    \lambda_1 b(I^{m}_1) + \lambda_2 b(I^{m}_2) + \sum_{k=3}^p \lambda_k b(I_k) + \lambda_1^2 A(I^{m}_1)+\lambda_2^2 A(I^{m}_2) + \sum_{k=3}^p \lambda_k^2 A(I_k) \\ & +
    \lambda_1 \lambda_2 A(I_1^{m},I_2^{m}) + \lambda_1 \sum_{k=3}^p \lambda_k A(I_1^{m},I_k) + \lambda_2 \sum_{k=3}^p \lambda_k A(I_2^{m},I_k) \\ & + \sum_{k=3}^{p-1} \sum_{k'=k+1}^{p} \lambda_k \lambda_{k'} A(I_k,I_{k'}) + \lambda_1 \lambda_2 A(i_m,j_m) = F(x^m)+ \lambda_1 \lambda_2 A(i_m,j_m).
\end{align*}
The inequality holds because of (\ref{eq:swaplem}), and the first and the last equalities follow from Lemma~\ref{lemma:swaplem}. The second to the last equality uses that $I_1^{m}=I_1^{m-1}$ and $I_2^{m} = I_2^{m-1} - j_m + i_m$.
\end{proof}

\begin{theorem}
\label{thm:swapthm}
Let $\mathcal{M}([n],\mathcal{I})$ be a matroid of rank $r$ and $P$ be its corresponding base polytope. Let $F(z)=\frac{1}{2}z^T A z + z^T b$ where $A, b\geq 0$ and $A$ is a symmetric matrix with zero diagonal that satisfies the $\sigma$-semi-metric inequality, i.e., $A(i,j) \leq \sigma (A(i,k)+A(j,k))$ for all $i,j,k\in [n]$. Let $f(S)=F(\mathbbm{1}_S)$ for any $S\subseteq [n]$. Let $x\in P$ and $S$ be the output of the modified swap rounding (Algorithm~\ref{alg:swapRounding}) on $x$. Then $F(x)\leq O(1+\frac{\sigma}{r}) f(S)$.
\end{theorem}
\begin{proof}
Let $x=\sum_{i=1}^p \lambda_i \mathbbm{1}_{I_i}\in P$ where $I_i$'s are bases of the matroid, $\sum_{i=1}^p \lambda_i = 1$, and $\lambda_i\geq 0$, for $i=1,\ldots,p$. Let $S$ be the output of the swap rounding (Algorithm~\ref{alg:swapRounding}) if it starts from $(I_1,\ldots,I_p)$ and $(\lambda_1,\ldots,\lambda_p)$. Let $x_k$ denote the vector corresponding to $\bm{I}_k = (I'_{k}, I_{k+1}, \ldots, I_p)$ and $\bm{\lambda}_k= (\lambda'_{k}, \lambda_{k+1}, \ldots, \lambda_p)$, i.e. $x_k= \lambda'_{k} \mathbbm{1}_{I'_k} + \sum_{i=k+1}^p \lambda_i \mathbbm{1}_{I_i}$. By Lemma~\ref{lemma:swaplem}, for $k=1,\ldots,n-1$, we have
\begin{equation*}
F(x_k) \leq F(x_{k+1})+\lambda_k' \lambda_{k+1} \sum_{(i,j)\in M_k} A(i,j) \leq F(x_{k+1})+\lambda_k' \lambda_{k+1} \sum_{(i,j)\in M_t} A(i,j),
\end{equation*}
where $t=\argmax_{k=1,\ldots,p-1} \{\sum_{(i,j)\in M_k} A(i,j)\}$. Therefore
\begin{align}
\label{eq:swapthm0}
\nonumber F(x_1) &\leq F(x_p)+(\sum_{k=1}^{p-1} \lambda'_{k} \lambda_{k+1} ) \sum_{(i,j)\in M_t} A(i,j) = F(x_p) + (\sum_{k=1}^{p-1} \sum_{m=1}^k \lambda_m \lambda_{k+1} ) \sum_{(i,j)\in M_t} A(i,j) \\ &
\leq F(x_p) + \frac{1}{2}\sum_{(i,j)\in M_t} A(i,j) = f(I_p')+\frac{1}{2}\sum_{(i,j)\in M_t} A(i,j),
\end{align}
where the last inequality holds since
$2 \sum_{k=1}^{p-1} \sum_{m=1}^k \lambda_m \lambda_{k+1} \leq (\sum_{k=1}^p \lambda_k)^2 = 1$.
Now, we bound the term $\sum_{(i,j)\in M_t} A(i,j)$. By definition of $M_t$, note that $M_t\subseteq I'_t\times I_{t+1}$. Using this and Lemma~\ref{lem:2ndordermulti} it follows that
\begin{equation}
\label{eq:swapthm1}
\sum_{(i,j)\in M_t} A(i,j) \leq A(I'_t,I_{t+1}) \leq 4 \cdot F(\frac{1}{2} \mathbbm{1}_{I'_t}+\frac{1}{2}\mathbbm{1}_{I_{t+1}}).
\end{equation}
By Lemma~\ref{lemma:swaplem} and the $\sigma$-semi-metric assumption, we also know that
\begin{align}
\label{eq:swapthm2}
F(\frac{1}{2}\mathbbm{1}_{I'_t}+\frac{1}{2}\mathbbm{1}_{I_{t+1}}) \leq F(\mathbbm{1}_{I^*})+ \frac{1}{4}\sum_{(i,j)\in M^*} A(i,j) \leq F(\mathbbm{1}_{I^*}) + \frac{1}{4}\sum_{(i,j)\in M^*} \frac{\sigma}{r-1} \big(A(i,I'_t-i)+A(j,I'_t-i) \big).
\end{align}
Note that none of the edges of $M^*$ is present in the right hand side summation. Therefore
\begin{align}
\label{eq:swapthm3}
 \nonumber \sum_{(i,j)\in M^*} (A(i,I'_t-i)+A(j,I'_t-i)) & \leq A(I'_t) + A(I'_t, I_{t+1}) - \sum_{(i,j)\in M^*} A(i,j) \\ & \leq 4 \cdot F(\frac{1}{2}\mathbbm{1}_{I'_t}+\frac{1}{2}\mathbbm{1}_{I_{t+1}}) - \sum_{(i,j)\in M^*} A(i,j) \leq 4 F(\mathbbm{1}_{I^*}) = 4 f(I^*).
\end{align}
where the second inequality follows from Lemma~\ref{lem:2ndordermulti} and the last inequality holds because of Lemma~\ref{lemma:swaplem}. 
Combining (\ref{eq:swapthm1}), (\ref{eq:swapthm2}), and (\ref{eq:swapthm3}), we get
\begin{align}
\label{eq:swapthm4}
\sum_{(i,j)\in M_t} A(i,j) \leq \big(4+\frac{4 \sigma}{r-1} \big) f(I^*).
\end{align}
Hence, by (\ref{eq:swapthm0}) and (\ref{eq:swapthm4}), we have
\[
F(x_1)\leq f(I'_p)+\Big(2+\frac{2 \sigma}{r-1} \Big) f(I^*) ,
\]
and this yields the result.
\end{proof}

\begin{algorithm}[H]
\RestyleAlgo{algoruled}
\caption{Swap rounding for monotone second-order-modular functions under matroid constraints}
\label{alg:swapRounding}
\footnotesize
\SetKwFunction{MB}{MergeBases}
    \SetKwProg{Fn}{Function}{:}{}
\textbf{Input:} A matroid $\mathcal{M}=([n],\mathcal{I})$, its base polytope $P$, and a fractional solution $x\in P$. A set function $f(S)=\sum_{i\in S} b(i)+\sum_{\{i,j\}\subseteq S} A(i,j)$. \\[0.6ex]
Find $\bm{\lambda}_1=(\lambda_1, \lambda_2,\ldots,\lambda_p)$ and $\bm{I}_1=(I_1,I_2,\ldots,I_p)$ such that $x=\sum_{i=1}^p \lambda_i I_i$, $\lambda_i\geq 0$ (for any $i$), $\sum_{i=1}^p \lambda_i = 1$, and $I_i$'s are bases of the matroid\;
$I_1'\leftarrow I_1$\;
$\lambda_1'\leftarrow \lambda_1$\;
\For{$k=1,\ldots,p-1$} {
    $(I_{k+1}', M_{k})\leftarrow$\MB{$\bm{I}_k$,$\bm{\lambda}_k$}\;
    $\lambda_{k+1}'\leftarrow \lambda_{k}'+\lambda_{k+1}$\;
    $\bm{I}_{k+1} \leftarrow (I_{k+1}',I_{k+2},\ldots,I_{p})$\;
    $\bm{\lambda}_{k+1} \leftarrow (\lambda_{k+1}',\lambda_{k+2},\ldots,\lambda_p)$\;
	}
    $t \leftarrow \argmax_{k=1,\ldots,p-1} \{\sum_{(i,j)\in M_k} A(i,j)\}$\;
    $(I^*,M^*)\leftarrow$\MB{$(I_t',I_{t+1})$,$(0.5,0.5)$}\;
\Return $\argmax \{f(I^*),f(I_p')\}$\;
 \BlankLine  \BlankLine
    \Fn{\MB{$\bm{I}=(I_1,I_2,\ldots,I_m) \;$  , $\; \bm{\lambda}=(\lambda_1,\lambda_2,\ldots,\lambda_m)$}}{
    $M \leftarrow \emptyset$\;
        \While{$I_1\neq I_2$}{
 Pick $i\in I_1 \setminus I_2$ and $j\in I_2 \setminus I_1$ such that $I_1-i+j\in\mathcal{I}$ and $I_2-j+i\in\mathcal{I}$\;
 $M\leftarrow M\cup \{(i,j)\}$\;
\If{$b(i)+\lambda_1 A(i, I_1-i)+\lambda_2 A(i, I_2-j) + \sum_{k=3}^{m} \lambda_k A(i, I_k) \geq b(j)+ \lambda_1 A(j,I_1-i) + \lambda_2 A(j,I_2-j) + \sum_{k=3}^{m} \lambda_k A(j, I_k)$}{$I_2\leftarrow I_2-j+i$\;}\Else{$I_1\leftarrow I_1-i+j$\;}}
 \Return $(I_1, M)$\;}
 \textbf{End Function}
\end{algorithm}

\section{Appendix: Hardness of Approximation for $\sigma$-Semi-Metric Diversity}
\label{focsapp:hardness}

In this section, we provide a hardness result for approximate maximization of $\sigma$-semi-metric diversity functions defined on a semi-metric distance. Our results are based on  inapproximability results for finding densest subgraphs.

Given a graph $G$ and integer $k$,  the {\em densest $k$-subgraph problem} aims to  find an induced subgraph of size $k$ with the maximum number of edges. Let $R$ be a subset of vertices of $G$ and $E(R)$ be the number of edges in the induced subgraph of $R$. The density of $R$ is defined as $\rho(R)=E(R)/{|R| \choose 2} \leq 1$.
A recent breakthrough \cite{manurangsi2017almost}  shows that, assuming  the exponential time hypothesis (ETH), there is no
subpolynomial approximation algorithm for densest subgraph. More precisely, there is no polytime algorithm which  can distinguish between
two cases: (i) an instance $G$ which contains a $k$-clique
and (ii) an instance where the density of every $k$-subset
$S$ satisfies $\rho(S) \leq n^{-1/(\log \log n)^c}$, where
$c$ is a universal constant independent of $n$. In the following, we let $\theta :=  n^{1/(\log \log n)^c}$.
Existence of constant-factor approximations had previously been ruled out under the unique games conjecture with small set expansion \cite{raghavendra2010graph}.

\begin{reptheorem}{thm:approxhardness}
Assuming ETH: (1)  There is no polytime $4/\theta$-approximation algorithm for maximizing  $\sigma$-semi-metric functions subject to a cardinality constraint, and (2)
for any fixed $\sigma \geq 1$ and $\epsilon >0$, there is no polytime algorithm which  approximates the maximum of a $\sigma$-semi-metric function subject to a cardinality constraint within a factor of $2\sigma-\epsilon$.
\end{reptheorem}
\begin{proof}

For $\sigma \geq 1$, we can reduce the densest $k$-subgraph problem to $\sigma$-semi-metric function maximization in the following way. Consider an instance of densest $k$-subgraph  on graph $G$ with vertex set $[n]$. Create a distance function $A:[n]\times [n]\rightarrow \mathbb{R}$. If there is an edge between $i,j\in [n]$ in $G$, set $A(i,j)=2\sigma>1$; otherwise set $A(i,j)=1$. It is easy to see that this distance function is $\sigma$-semi-metric.  Let $f(R)=\sum_{\{i,j\}\subseteq R} A(i,j)$. If $|R|=k$, we have
\[
f(R) = 2\sigma E(R) + ({k \choose 2}-E(R)).
\]
We know ${k \choose 2}\geq E(R)$. Therefore
\begin{equation*}
2\sigma E(R) \leq f(R) \leq 2\sigma E(R) + {k \choose 2},
\end{equation*}
and dividing both sides by $2\sigma {k \choose 2}$ we get
\begin{equation}
\label{eqn:densesubhard}
\rho(R)\leq \frac{f(R)}{2\sigma {k\choose 2}}\leq \rho(R) + \frac{1}{2\sigma}.
\end{equation}
It is also easy to see that
\[
\argmax_{\substack{R\subseteq [n]\\|R|=k}}\rho(R) = \argmax_{\substack{R\subseteq [n]\\|R|=k}} f(R).
\]

 Suppose there is   a $\frac{4}{\theta}$-approximation algorithm for  maximizing $\sigma$-semi-metric functions.
Let its output on $G$ be $S$ and choose
\[
\texttt{OPT}\in \argmax_{\substack{R\subseteq [n]\\|R|=k}}\rho(R).
\]
We have
\[
\rho(\texttt{OPT})\leq \frac{f(\texttt{OPT})}{2\sigma {k\choose 2}}\leq \frac{(\theta/4) f(S)}{2 \sigma {k\choose 2}}\leq (\theta/4) \rho(S)+ \frac{\theta}{8\sigma}.
\]
We can choose our $\sigma=\sigma(n) \geq \theta$  so that
 $\frac{\theta}{8\sigma} \leq \frac{1}{2}$. Hence $\rho(\texttt{OPT})\leq (\theta/4)\rho(S)+\frac{1}{2}$.  If $G$ is a graph in which the density of every subset of vertices of size $k$ is at most $1/\theta$, then clearly $\rho(S)\leq 1/\theta$.  If $G$ is a graph that contains a clique of size $k$, then $1=\rho(\texttt{OPT})\leq (\theta/4)\rho(S)+\frac{1}{2}$, and so $\rho(S)\geq \frac{2}{\theta}$. This means that our $1/\theta$-approximation algorithm can distinguish between these two graphs,  contradicting the implications from \cite{manurangsi2017almost}.

For (2), consider a given  $\sigma \geq 1, \epsilon >0$ and suppose there is a $(2\sigma - \epsilon)$-factor approximate algorithm for maximizing a $\sigma$-semi-metric function. Denote its output on $G$ by $S$, and let $OPT$ be defined as above. We then have
\[
\rho(\texttt{OPT})\leq \frac{f(\texttt{OPT})}{2\sigma {k\choose 2}}\leq \frac{(2\sigma - \epsilon)f(S)}{2\sigma {k\choose 2}}\leq (2\sigma - \epsilon) \rho(S)+\frac{2\sigma - \epsilon}{2\sigma}.
\]
Set $\delta=(\frac{1}{2\sigma -\epsilon}-\frac{1}{2\sigma})/2 = \frac{\epsilon}{4\sigma (2\sigma - \epsilon)}$, and note that $\delta > 0$ is a constant. If $G$ is a graph in which the density of every subset of vertices of size $k$ is at most $1/n^{\frac{1}{\log \log n)^c}}$, then clearly $\rho(S)\leq 1/n^{\frac{1}{\log \log n)^c}}$ and this is at most $\delta$ for $n$ sufficiently large. If $G$ is a graph that contains a clique of size $k$, then $1=\rho(\texttt{OPT})\leq (2\sigma - \epsilon) \rho(S)+\frac{2\sigma - \epsilon}{2\sigma}$ which means $\rho(S)\geq \frac{1}{2\sigma - \epsilon}-\frac{1}{2\sigma} = 2 \delta$. This means that our $(2\sigma - \epsilon)$-factor approximate algorithm can distinguish between these two graphs which again contradicts the implications of \cite{manurangsi2017almost}.
\end{proof}

\old{
\begin{reptheorem}{thm:approxhardness}
Assuming the Planted Clique Conjecture: (1) for any constant $\sigma \geq 1$, it is hard to approximate the maximum of a $\sigma$-semi-metric function subject to a cardinality constraint within a factor of $2\sigma-\epsilon$ for any $\epsilon > 0$ and (2)
for a super-constant $\sigma$, there is no constant factor (polytime) approximation algorithm for maximizing a $\sigma$-semi-metric function subject to a cardinality constraint.
\end{reptheorem}
\begin{proof}
Planted Clique problem asks for an algorithm to distinguish between the following graphs with probability of at least $3/4$: 1) A graph drawn from $\mathcal{G}(n,1/2)$, 2) A graph drawn from $\mathcal{G}(n,1/2)$ and then a clique of size $n^{1/2-\delta}$ is planted in it ($\delta>0$)~\cite{karp1976probabilistic}. The planted clique conjecture states that there is no polynomial time algorithm to do this task~\cite{alon2011inapproximability, Jerrum92}. It has been shown that assuming the planted clique conjecture, it is hard to approximate the maximum of a metric diversity function within a factor better than $2$~\cite{BhaskaraGMS16,BorodinLY12}.

Given a graph $G$, in the densest $k$-subgraph problem we need to find an induced subgraph of size $k$ with the maximum number of edges. Let $R$ be a subset of vertices of $G$ and $E(R)$ be the number of edges in the induced subgraph of $R$. The density of $R$ is defined as $\rho(R)=E(R)/{|R| \choose 2}$. Alon et al.~\cite{alon2011inapproximability} showed that if there is no polynomial time algorithm for the planted clique problem for a planted clique of size $n^{1/3}$, then there is no polynomial time algorithm for distinguishing between a graph $G_1$ of size $n$ that contains a clique of size $n^{1/3}$, and a graph $G_2$ of the same size in which the density of every subset of vertices of size $n^{1/3}$ is at most $\delta>0$.

For $\sigma \geq 1$, we can reduce the densest $k$-subgraph problem to $\sigma$-semi-metric function maximization in the following way. Consider an instance of densest $k$-subgraph ($k=n^{1/3}$) on graph $G$ with vertex set $[n]$. Create the distance function $A:[n]\times [n]\rightarrow \mathbb{R}$. If there is an edge between $i,j\in [n]$ in $G$, set $A(i,j)=2\sigma$, otherwise set $A(i,j)=1$. It is easy to see that this distance function is $\sigma$-semi-metric.  Let $f(R)=\sum_{\{i,j\}\subseteq R} A(i,j)$. If $|R|=k$, we have
\[
f(R) = 2\sigma E(R) + ({k \choose 2}-E(R)).
\]
We know ${k \choose 2}\geq E(R)$. Therefore
\begin{equation*}
2\sigma E(R) \leq f(R) \leq 2\sigma E(R) + {k \choose 2},
\end{equation*}
and dividing both sides by $2\sigma {k \choose 2}$ we get
\begin{equation}
\label{eqn:densesubhard}
\rho(R)\leq \frac{f(R)}{2\sigma {k\choose 2}}\leq \rho(R) + \frac{1}{2\sigma}.
\end{equation}
It is also easy to see that
\[
\argmax_{\substack{R\subseteq [n]\\|R|=k}}\rho(R) = \argmax_{\substack{R\subseteq [n]\\|R|=k}} f(R).
\]
Now, assume that for some fixed constant $\theta\geq 1$ there is a $\frac{1}{\theta}$-factor approximate algorithm for finding the maximum of $\sigma$-semi-metric function ($\sigma$ is super-constant) and its output on $G$ is $S$. Also, let
\[
\texttt{OPT}\in \argmax_{\substack{R\subseteq [n]\\|R|=k}}\rho(R).
\]
We have
\[
\rho(\texttt{OPT})\leq \frac{f(\texttt{OPT})}{2\sigma {k\choose 2}}\leq \frac{\theta f(S)}{2\sigma {k\choose 2}}\leq \theta\rho(S)+ \frac{\theta}{2\sigma}
\]
Since $\sigma \in \omega(1)$, for some $n$ large enough we have that $\frac{\theta}{2\sigma} \leq \frac{1}{2}$. Hence $\rho(\texttt{OPT})\leq \theta\rho(S)+\frac{1}{2}$. Set $\delta=\frac{1}{4\theta}$ and note that $\delta>0$ is a constant. If $G$ is a graph in which the density of every subset of vertices of size $k$ is at most $\delta$ then clearly $\rho(S)\leq \delta$. If $G$ is a graph that contains a clique of size $k$ then $1=\rho(\texttt{OPT})\leq \theta\rho(S)+\frac{1}{2}$, which means $\rho(S)\geq \frac{1}{2\theta} = 2 \delta$. This means that our $\theta$-factor approximate algorithm can distinguish between these two graphs which is in contrast with the planted clique conjecture and Alon et al. result.

For the first part, given any constant $\sigma$,
assume there is a $(2\sigma - \epsilon)$-factor approximate algorithm for some $\epsilon >0$ for finding the maximum of $\sigma$-semi-metric function. Denote its output on $G$ by $S$, and let $OPT$ be defined as above. We then have
\[
\rho(\texttt{OPT})\leq \frac{f(\texttt{OPT})}{2\sigma {k\choose 2}}\leq \frac{(2\sigma - \epsilon)f(S)}{2\sigma {k\choose 2}}\leq (2\sigma - \epsilon) \rho(S)+\frac{2\sigma - \epsilon}{2\sigma}.
\]
Set $\delta=(\frac{1}{2\sigma -\epsilon}-\frac{1}{2\sigma})/2 = \frac{\epsilon}{4\sigma (2\sigma - \epsilon)}$, and note that $\delta > 0$ is a constant. If $G$ is a graph in which the density of every subset of vertices of size $k$ is at most $\delta$ then clearly $\rho(S)\leq \delta$. If $G$ is a graph that contains a clique of size $k$ then $1=\rho(\texttt{OPT})\leq (2\sigma - \epsilon) \rho(S)+\frac{2\sigma - \epsilon}{2\sigma}$ which means $\rho(S)\geq \frac{1}{2\sigma - \epsilon}-\frac{1}{2\sigma} = 2 \delta$. This means that our $(2\sigma - \epsilon)$-factor approximate algorithm can distinguish between these two graphs which is in contrast with the planted clique conjecture and Alon et al. result.
\end{proof}
}

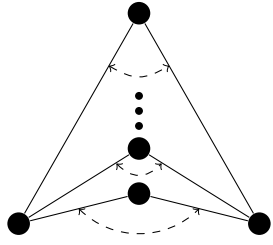
\begin{figure}
    \centering
    \centering
\begin{tikzpicture}[scale=0.4]
\tikzstyle{s}=[circle,fill=black,minimum size=0.3cm,inner sep=0pt]
\tikzstyle{t}=[circle,fill=black,minimum size=0.15cm,inner sep=0pt]
	\begin{pgfonlayer}{nodelayer}
		\node [s] (0) at (-4, -4) {};
		\node [s] (1) at (4, -4) {};
		\node [s] (2) at (0, -3) {};
		\node [s] (3) at (0, -1.5) {};
		\node [s] (4) at (0, 3) {};
		\node [style=none] (5) at (-2, -3.5) {};
		\node [style=none] (6) at (2, -3.5) {};
		\node [style=none] (7) at (0.75, -2) {};
		\node [style=none] (8) at (-0.75, -2) {};
		\node [style=none] (9) at (-1, 1.25) {};
		\node [style=none] (10) at (1, 1.25) {};
		\node[scale=0.7] [t] (11) at (0, -0.75) {};
		\node[scale=0.7] [t] (12) at (0, -0.25) {};
		\node[scale=0.7] [t] (13) at (0, 0.25) {};
	\end{pgfonlayer}
	\begin{pgfonlayer}{edgelayer}
		\draw (0) to (2);
		\draw (2) to (1);
		\draw (1) to (3);
		\draw (3) to (0);
		\draw (0) to (4);
		\draw (4) to (1);
		\draw [style=double arrow dotted, bend right, looseness=1.25] (9.center) to (10.center);
		\draw [style=double arrow dotted, bend right=45] (5.center) to (6.center);
		\draw [style=double arrow dotted, bend right=45, looseness=1.25] (8.center) to (7.center);
	\end{pgfonlayer}
\end{tikzpicture}
    \caption{Lower bound on integrality gap for quadratic functions.}
    \label{fig:intgaplowerbound}
\end{figure}

\begin{reptheorem}{thm:intgaplower}
Let $k,t\in\mathbb{N}$ with $1 \leq t \leq k$. There exists a $\sigma$-semi-metric with multilinear extension $F$, and a matroid $\mathcal{M}=([2k],\I)$ with rank $r=k+t-1$ and minimum circuit size $c = 2t$, where the integrality gap of $F(x)$ over the matroid polytope $P_\mathcal{M}$ is $\Omega(\min \{ \frac{r}{c-2},\frac{\sigma}{r}\})$.
\end{reptheorem}
\begin{proof}
Let $S_i=\{2i-1,2i\}$ for $1\leq i\leq k$, and $\mathcal{S} = \{ S_1, S_2, \ldots, S_k\}$. We define a matroid $\mathcal{M}=([2k],\I)$ in terms of its circuits as follows. A set $C$ is a circuit of  $\mathcal{M}$ if and only if $C$ is the union of any $t$ sets $S_i$. It is then clear that the minimum size $c$ of a circuit is $2t$, and the rank $r$ of the matroid is $k+t-1$. 
For example, $\mathcal{M}$ could be the graphic matroid corresponding to the graph in Figure~\ref{fig:intgaplowerbound}. Circuits here correspond to cycles of size $4$, and the dashed lines show the non-zero coefficients of $F$.

Let $F(x) = \sum_{\{u,v\} \in \mathcal{S} } x_u x_v + \sum_{\{u,v\} \notin \mathcal{S} } \frac{1}{\sigma} x_u x_v$. It is straightforward to see that $F$ is the multilinear extension of a $\sigma$-semi-metric diveristy function induced by a complete graph which has weight $1$ on edges from $\mathcal{S}$ and weight $1/\sigma$ otherwise.

By definition of $\mathcal{M}$ and $F$, it is clear that any integral solution $x_I \in P_\mathcal{M}$ maximizing $F$ will pick $t-1$ pairs from $\mathcal{S}$ and then singletons from other pairs.
Therefore
\begin{align*}
F(x_I) & :=\max_{x\in P_\mathcal{M} \cap \{0,1\}^{2k}} F(x) = (t-1) + \frac{1}{\sigma} \Big(  {r \choose 2} - (t-1) \Big)   = (1-\frac{1}{\sigma})(t-1)+\frac{1}{\sigma} {r \choose 2} \\ & = \frac{ (\sigma-1)(c-2) + r (r-1) }{2 \sigma}.
\end{align*}
On the other hand, $x_0 := \frac{k+t-1}{2k}\mathbbm{1}_{[2k]} \in P_\mathcal{M}$ 
and
$$
F(x_0)= k (\frac{k+t-1}{2k})^2 + \Big({2k \choose 2} - k  \Big)  \frac{1}{\sigma} (\frac{k+t-1}{2k})^2 =  k \big(\frac{k+t-1}{2k}\big)^2 \big(1 + \frac{2 (k-1)}{\sigma }  \big).
$$
Using that $r = k+t-1$ and $k = r-\frac{c}{2}+1$ we have
$$
    k(\frac{k+t-1}{2k})^2 = \frac{r^2}{4(r-\frac{c}{2}+1)} = \frac{r^2}{2(2r-c+2)} \geq \frac{r}{4},
$$
where the last inequality follows since $c\geq 2$.
Hence,
$
F(x_0) \geq \frac{r}{4} (1 + \frac{2 (k-1)}{\sigma } ).
$
It follows that the integrality gap is at least
$$
\frac{F(x_0)}{F(x_I)} \geq \frac{1}{2} \cdot \frac{\sigma r + 2 r (k-1) }{(\sigma - 1)(c-2) + r (r-1)}
\geq \frac{1}{2} \cdot \frac{\sigma r }{\sigma (c-2) + r^2}
\geq \frac{1}{4} \cdot \min\{ \frac{r}{c-2}, \frac{\sigma}{r} \}.
$$
\end{proof}

\section{Appendix: One-Sided Smoothness versus Lipschitz Smoothness}
\label{focsapp:lipschitz}
Lipschitz smoothness is an important, widely-used property in convex optimization and machine learning.
One-sided $\smooth$-smoothness  is different from Lipschitz smoothness (and other smoothness notions based on  Holder's or uniform continuity) and we believe it may also have applications to these areas.

A differentiable function is {\em  Lipschitz smooth} if its gradient is Lipschitz continuous. In other words, $f$ is Lipschitz smooth if there exists $L\geq 0$ such that for any $x$ and $y$,
$
||\nabla f(x)- \nabla f(y)||_2 \leq L||x-y||_2
$
or equivalently for twice differentiable functions,
$
u^T \nabla^2 f(x) u \leq L||u||_2^2.
$
We then call  $f$  {\em  $L$-Lipschitz smooth}. One could define the one-sided version of this smoothness if the above inequality holds for any $x\leq y$ (second definition/inequality holds for any $u\geq \vec{0}$). With this definition, it is easy to see that submodular functions are one-sided $0$-Lipschitz smooth. On the other hand one-sided $\smooth$-smoothness is not equivalent to one-sided $L$-Lipschitz smoothness. To see an important  difference, consider $g=cf$ function where $c$ is a constant and $f$ is one-sided smooth. We have $\nabla g(x) = c\nabla f(x)$. Thus if  $f$ is one-sided $L$-Lipschitz smooth we may only assert that $g$ is one-sided $cL$-Lipschitz smooth. In particular,  Lipschitz smoothness is not closed under multiplication. On the other hand,  the one-sided $\smooth$-smooth functions form a cone. Intuitively, the reason is that in $\smooth$-smooth functions, the ratio of the gradients is bounded (as shown in Lemma~\ref{lemma:epsilonchangegradient}) unlike Lipschitz smoothness where the difference of the gradients is bounded.

\section{Appendix: Other Applications}

\subsection{Appendix: The Diversified Procurement Problem}
\label{app:procure}

Consider a problem whereby an organization decides how to outsource the building or servicing of a system to a collection of $n$ competing  vendors. The outcome  is an allocation of  work across the vendors. An allocation is
represented by a vector $x \in P \subseteq \mathbb{R}^n$; we focus on the case where $P \subseteq [0,1]^n$. Possibly $P$ is just $P=\{x: ||x||_1 \leq k\}$ but may also incorporate structural constraints imposed by the system or to enforce a  resilience solution (e.g., avoid allocations where one vendor becomes  too big to fail). 
Given bids $b_i$ for each $i$, the  payoff  to the organization is $b^T x$.  A different type of consideration for the procuring organization is to  build  diversity into the  work-plan which results. We consider two sources for lack of diversity.  First, there may be   collusions  and these  are to be subdued.  The organization can define
 a matrix  $[A_{ij}] \geq 0$ which  estimates  pairwise collusions. A solution which lessens the value of $x^T A x$ is more desired.
Second, the system may  be serving a collection of $m$ stakeholder communities. Different vendors may be more desirable than others to
distinct communities. Again, the procuring organization can model this by defining   vectors $g_i \in \mathbb{R}^m$, where $g_i(j) \in \mathbb{R}$ represents the level of support (positive or negative) it receives from community $j$. The overall measure of quality seeks  a solution which promotes representation across more communities (vectors with $g_i^T g_{i'} < 0$ are pointing in different directions; hence  good). We propose the following model to address this multi-criteria objective: 
\[
F(x) = \frac{1}{2} x^T (-A + - G^T G) x  + b^T x,
\]

\noindent
where $G$ is the $m \times n$ matrix whose columns are $g_i$. 
Hence $F$ consists of a revenue part $b^Tx$ and  a penalty part for lack of diversity.

It is easily
seen that $A+G^T G$ is copositive and hence $F(x)$ is $0$-\OSS --- see Section~\ref{sec:zero}.\footnote{This also leads to  examples of $0$-\OSS functions whose Hessians have positive off-diagonals. For instance,  by taking $A=0$ and select  vectors $g_i$ which are pairwise oblique (i.e., $g_i \cdot   g_j < 0$).}  Hence
 the jump-start continuous greedy process of Section~\ref{sec:cont-greedy-and-smoothness} can be applied if the model has been defined so that $\nabla F(x) \geq 0$ (since $F$ is normalized this ensures non-negativity). Checking this gradient condition is easy. Moreover, it  is useful in the modelling phase for  exploring the trade-offs between the vector $b$ of bids and the community representation objectives which are  determined by $G$. 
 
 
Note that Theorem~\ref{thm:thirdderiv} implies that the continuous greedy process produces a solution which is within a factor  $1-1/e$ 
of the  fractional optimum of
 $\max\{ F(x): x \in P\}$. We may also create
 a (weakly) polytime $1-1/e$ approximation as follows. 
 
 We assume that the model has been constructed so that the bids dominate the
 diversity penalties. Concretely we assume that $\nabla F(x)=Cx + b > \vec{\frac{1}{n}}$.  We also let $U$ be an upper bound on the entries of $C:=A+G^TG$, i.e., $\max |C_{ij}|$.
 We now examine for which values $\eta$ is  the function $F(x)$  $\eta$-local --- see paragraph before Theorem~\ref{thm:discretize_eta}.
That is, we want:
 \[
 u^T \nabla F(x+\epsilon u) \geq (1-\eta \epsilon) u^T \nabla F(x)
 \]
 
 \noindent
 for any $x,u \geq 0$, $\epsilon \in [0,1]$. This holds if we have:
  \[
 u^T (Cx+b) + \epsilon u^T C u  \geq (1-\eta \epsilon) (Cx+b).
 \]
 
 \noindent 
 By re-arranging this holds as long as
  \begin{equation}
      \label{eqn:eigs}
 - u^T C u \leq \eta u^T (Cx+b)
 \end{equation}

\noindent
 We now use the fact
 that $|u^T C u| \leq \lambda_{|max|} ||u||^2_2$, where $\lambda_{|max|}$ is the maximum value of $|\lambda|$ for an eigenvalue $\lambda$ of $C$.
 By  the Gershgorin Circle Theorem \cite{gershgorin1931uber} 
$\lambda_{|max|}$  is
 at most $|C_{ii}|+ \sum_{j \neq i} |C_{ij}| \leq n U$.
 Hence $|u^T C u| \leq (n U) ||u||^2_2$
 which is at most $n^2 U ||u||_1$ if $||u||_{\infty} \leq 1$.  By our gradient assumption,  the right hand side of (\ref{eqn:eigs}) is then at least $\frac{\eta}{n} ||u||_1$. Note that 
 we only need to establish $\eta$-locality for vectors $u=v_{Max}$ selected by the greedy process. Since these vectors lie in $P$ which  is in the unit hypercube, we have the desired inequality.
 Hence we may choose $\eta = n^3 U$ and
 a discretization follows from
 Theorem~\ref{thm:discretize_eta}.

\end{document}